\documentclass[journal,comsoc, 12pt,onecolumn,draftclsnofoot]{IEEEtran}

\usepackage{textcomp}
\usepackage{verbatim}

\usepackage{tikz}
\usepackage[T1]{fontenc}
\usetikzlibrary{patterns, arrows,backgrounds,fit,tikzmark,positioning,calc,decorations,snakes,shapes,matrix}
\usetikzlibrary{decorations.pathreplacing}
\usetikzlibrary{calc,arrows.meta,positioning}

\usepackage[T1]{fontenc}
\usepackage{lmodern}

\usepackage{cite}

\usepackage{amsfonts}
\usepackage{amsmath,bm, amssymb}
\usepackage{mathtools,hyperref}
\usepackage{amsthm}
\usepackage{mathrsfs}

\usepackage{xcolor}

\usepackage{subcaption} 
\usepackage{nicefrac} 
\usepackage{array}
\usepackage{pgfplots}
\usepackage{pgf,tikz}
\usetikzlibrary{arrows}

\usepackage[subnum]{cases} 
\usepackage[inline]{enumitem}
\usepackage{lscape}
\usepackage{cleveref}
\usepackage{multicol}
\usepackage{multirow}	

\usepackage{colortbl}
\usepackage{adjustbox}

\theoremstyle{remark}
\newtheorem{theorem}{ {Theorem}}

\newtheorem{definition}{{Definition}}

\newtheorem{lemma}{ {Lemma}}
\newtheorem{remark}{ {Remark}}

\definecolor{RedBlue}{rgb}{0.8,0,0.5}
\definecolor{RedBlueGreen}{rgb}{0.8,0.6,0.5}
\definecolor{YellowOrange}{rgb}{0.4,0.4,0}
\definecolor{OliveGreen}{rgb}{0,0.6,0}

\usetikzlibrary{decorations.text}
\usepgfplotslibrary{fillbetween}
\usetikzlibrary{calc, fit}
\usetikzlibrary{arrows.meta}
\usetikzlibrary{patterns}
\usetikzlibrary{fit,matrix}

\ifCLASSINFOpdf
\else
\fi
\usepackage{amsmath}
\usepackage[cmintegrals]{newtxmath}
\begin{document}
	
\title{Rate-Efficiency and Straggler-Robustness through Partition in Distributed Two-Sided Secure Matrix Computation}

\author{\IEEEauthorblockN{\small{Jaber~Kakar$^{*}$,~\IEEEmembership{\small Student Member,~IEEE,}
        Seyedhamed~Ebadifar$^{*}$,
        \\and~Aydin~Sezgin$^{*}$,~\IEEEmembership{\small Senior~Member,~IEEE}\\}
	\IEEEauthorblockA{$^{*}$Institute of Digital Communication Systems,
		Ruhr-Universit{\"a}t Bochum, Germany \\
		Email: \{jaber.kakar, seyedhamed.ebadifar, aydin.sezgin\}@rub.de}
	}        
}  

\markboth{Draft}%
{Shell \MakeLowercase{\textit{et al.}}: Bare Demo of IEEEtran.cls for IEEE Communications Society Journals}

\makeatletter
\newcommand*{\rom}[1]{\expandafter\@slowromancap\romannumeral #1@}
\makeatother

\maketitle

\newcommand{\alert}[1]{\textcolor{black}{#1}}
\newcommand{\alertv}[1]{\textcolor{black}{#1}}

\begin{abstract}
	Computationally efficient matrix multiplication is a fundamental requirement in various fields, including and particularly in data analytics. To do so, the computation task of a large-scale matrix multiplication is typically outsourced to multiple servers. However, due to data misusage at the servers, security is typically of concern. In this paper, we study the two-sided secure matrix multiplication problem, where a user is interested in the matrix product $\bm A\bm B$ of two finite field private matrices $\bm A$ and $\bm B$ from an information-theoretic perspective. In this problem, the user exploits the computational resources of $N$ servers to compute the matrix product, but simultaneously tries to conceal the private matrices from the servers. Our goal is twofold: (i) to \emph{maximize} the communication rate, and, (ii) to \emph{minimize} the effective number of server observations needed to determine $\bm A\bm B$, while preserving security, where we allow for up to $\ell\leq N$ servers to collude. To this end, we propose a general aligned secret sharing scheme for which we optimize the matrix partition of matrices $\bm A$ and $\bm B$ in order to either optimize objective (i) or (ii) as a function of the system parameters (e.g., $N$ and $\ell$). A proposed \emph{inductive} approach gives us \emph{analytical, close-to-optimal} solutions for both (i) and (ii). With respect to (i), our scheme significantly outperforms the existing scheme of \emph{Chang and Tandon} in terms of (a) communication rate, (b) maximum tolerable number of colluding servers and (c) computational complexity.    
\end{abstract}


\IEEEpeerreviewmaketitle

\section{Introduction}
\label{sec:intro}
In machine learning and scientific computation, matrix multiplication plays an important role. \alert{However,} in many cases high memory requirements and computational effort is \alert{required}. Distributed approaches have been used to \alert{circumvent} computational and memory related barriers of matrix multiplication \cite{BOWLER2001255, x2, x3, x4}. Although distributed matrix multiplication can resolve computational and memory related difficulties, it causes new security problems. In the cryptography literature, different schemes have been proposed that balance security and efficiency of distributed matrix multiplication. \emph{Bultel et al.} \cite{laskdjfdfsadlkj} \alert{suggest partially homomorphic encryption approaches in the framework of MapReduce matrix multiplication.} \alert{In other related works, cryptographic techniques are applied to the problem of distributed matrix multiplication in cloud computing \cite{7841783,6755190}.}

As opposed to cryptographic techniques, \emph{information-theoretic} techniques have been hardly applied to the problem of secure matrix multiplication. In \cite{8437651}, \emph{\alert{Nodehi and Maddah-Ali}} apply \alert{information theory to the framework of \emph{limited-sharing multi-party computation}}. In limited-sharing multi-party computation, a set of sources offload the computation task, i.e., computing a polynomial function of input matrices available at the sources, to a set of \alert{servers}. The result of the computation has to be delivered to a master node. \alert{The authors} propose an efficient \alert{\emph{polynomial sharing}} scheme that minimizes the number of required servers (which is known as \emph{recovery threshold}) while preserving the privacy of \alert{colluding} servers and the master. Similar schemes have been applied to the context of non-uniform computation delays at the servers \cite{Yu2018StragglerMI}. 

\alert{In \cite{Chang2018OnTC}, \emph{Chang and Tandon} study the communication rate of a secure matrix multiplication problem consisting of a single user and $N$ \alert{curious} servers which are responsible for the computation of two matrices available at the user. The communication rate is seeked to be maximized when $\ell$ servers collude. The authors divide the security problem in two models.} 
\begin{enumerate}[label=($\roman*$)]
	\item \textbf{One-sided:} \alert{Only one of the two matrices is private. The other matrix is publicly available at all servers.}   
	\item \textbf{Two-sided:} Both matrices stored at the user are private and not available at the servers. 
\end{enumerate}
\alert{While for the \emph{one-sided model}, they characterize the capacity with respect to the communication rate, the capacity for the \emph{two-sided model} remains unknown. By comparing with the converse, their proposed scheme for the second model seems to be loose in terms of communication rate and the maximum number of tolerable colluding servers supporting a non-zero rate.}

\alert{In this paper, we propose a novel \emph{aligned secret sharing} scheme under arbitrary matrix partition for the two-sided model to optimize the two conflicting metrics -- (i) rate and (ii) recovery threshold. To this end, we formulate two optimization problems, (i) one which \emph{maximizes} the rate and (ii) the other which \emph{minimizes} the number of effective server needed when computing $\bm A\bm B$ subject to a minimum rate constraint. Both optimization problems find the optimal matrix partition of the matrices $\bm A$ and $\bm B$. Through an inductive approach, we find \emph{analytical, close-to-optimal} solutions of the optimization problems. These solutions identify the optimal matrix partition as a function of $N$ and $\ell$ and a minimum rate requirement $R_{\text{th}}$. With respect to objective (i), our scheme significantly improves upon the scheme of Chang and Tandon in terms of rate, computational complexity on the servers and the maximum number of tolerable colluding servers. While the maximum number of tolerable colluding servers of the scheme proposed by Chang and Tandon is equal to $\lfloor \sqrt{N}-1\rfloor$, our scheme attains a non-zero rate for up to $\lfloor \nicefrac{(N-1)}{2} \rfloor$ colluding servers. Despite of the higher communication rate in comparison to \cite{Chang2018OnTC}, our scheme attains a lower computational complexity at the servers.}

\textbf{Notations:} Throughout this paper, boldface lower-case and capital letters represent vectors and matrices, respectively. Further, for any two integers $a,b$ with $a\leq b$, we define $[a:b]\triangleq\{a,a+1,\ldots,b\}$ and we denote $[1:b]$ simply as $[b]$. $\mathbb{Z}$ refers to the set of all integers, while $\mathbb{Z}^{+}$ to the subset of positive integers. 

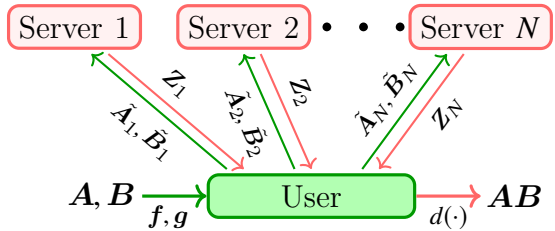
\begin{figure}[t]
	\centering
	\begin{tikzpicture}[roundnode/.style={circle, draw=green!60, fill=green!5, very thick, minimum size=7mm}, squarednode/.style={rectangle, draw=red!60, fill=red!5, very thick, minimum size=5mm, rounded corners},scale=0.9]
	
	\node[squarednode,fill=green!30, draw=OliveGreen] (u) at (3, 0) {$\qquad\text{User}\qquad$};
	\node[squarednode] (s1) at (0-0.5, 2.5) {Server $1$};
	\node[squarednode] (s2) at (2, 2.5) {Server $2$};
	\node[scale=2.5] (dots) at (3.75,2.5) {$\ldots$};
	\node[squarednode] (s2) at (5.5, 2.5) {Server $N$};
	
	\draw[->, thick, OliveGreen] (1.75,0.4) -- (-0.25,2.1) node[pos=0.5,sloped,below, black] {\footnotesize $\tilde{\bm A}_{1},\tilde{\bm B}_{1}$};
	\draw[<-, thick, red!60] (2,0.4) -- (0,2.1) node[pos=0.6,sloped,above, black] {\footnotesize $\bm Z_{1}$};
	
	\draw[->, thick, OliveGreen] (2.75,0.4) -- (2,2.1) node[pos=0.5,sloped,below, black] {\footnotesize $\tilde{\bm A}_{2},\tilde{\bm B}_{2}$};
	\draw[<-, thick, red!60] (3,0.4) -- (2.25,2.1) node[pos=0.6,sloped,above, black] {\footnotesize $\bm Z_{2}$};
	
	\draw[->, thick, OliveGreen] (3.75,0.4) -- (5,2.1) node[pos=0.5,sloped,above, black] {\footnotesize $\tilde{\bm A}_{N},\tilde{\bm B}_{N}$};
	\draw[<-, thick, red!60] (4,0.4) -- (5.25,2.1) node[pos=0.6,sloped,below, black] {\footnotesize $\bm Z_{N}$};	
	
	\node (inp) at (-0.1,0) {$\bm A,\bm B$};
	\draw[->, very thick, OliveGreen] (0.5,0) -- (1.45,0) node[pos=0.4,sloped,below, black, scale=0.8] {$\bm f,\bm g$};
	\draw[->, very thick, red!60] (4.55,0) -- (5.5,0) node[pos=0.5,sloped,below, black, scale=0.8] {$d(\cdot)$};;	
	\node (out) at (6,0) {$\bm A\bm B$};
	\end{tikzpicture}
	\caption{\footnotesize System model of the two-sided distributed matrix multiplication problem.}	
	\label{fig:SymMod}
\end{figure}

\section{System Model}
\label{sec:sym_mod}

In a fully, or two-sided, secure matrix multiplication problem, a user is interested in computing the matrix product $\boldsymbol{A}\boldsymbol{B}$ of two \emph{private} matrices $\boldsymbol{A}\in\mathbb{F}^{m\times n}$ and $\boldsymbol{B}\in\mathbb{F}^{n\times p}$\footnote{Each element is from a sufficiently large field $\mathbb{F}$.} securely (see Fig. \ref{fig:SymMod}). Hereby, the user has access to a distributed computation system consisting of $N$ honest, but curious computation servers connected to the user by private, error-free links. The user seeks the support of these servers but aims at concealing $\boldsymbol{A}$ and $\boldsymbol{B}$ from the servers. 

To this end, the user deploys encoding functions $f_i$ and $g_i$ to generate securely encoded matrices
$\tilde{\boldsymbol{A}}_i=f_i(\boldsymbol{A})$ and $\tilde{\boldsymbol{B}}_i=g_i(\boldsymbol{B})$ which are sent to the $i$-th server. The set of all $N$ encoding functions with respect to matrices $\boldsymbol{A}$ and $\boldsymbol{B}$ are denoted by $\boldsymbol{f}=(f_1,\ldots,f_N)$ and $\boldsymbol{g}=(g_1,\ldots,g_N)$, respectively.

Since every server $i$ is by assumption honest, the answer of the $i$-th server denoted by $\boldsymbol{Z}_i$ is a \emph{deterministic} function\footnote{This function is known by the user.} of $\tilde{\boldsymbol{A}}_i$ and $\tilde{\boldsymbol{B}}_i$, i.e.,
\begin{align*}
	H(\bm Z_i|\tilde{\boldsymbol{A}}_i,\tilde{\boldsymbol{B}}_i)=0.
\end{align*}
The user has to be able to determine $\boldsymbol{A}\boldsymbol{B}$ after applying the decoding function $d(\cdot)$ on the collection of all $N$ answers $\boldsymbol{Z}_1,\ldots,\boldsymbol{Z}_N$. i.e., $\boldsymbol{A}\boldsymbol{B}=d(\boldsymbol{Z}_1,\ldots,\boldsymbol{Z}_N)$, or information-theoretically satisfy the \emph{decodability constraint} 
\begin{align}
H(\boldsymbol{A}\boldsymbol{B}|\boldsymbol{Z}_1,\ldots,\boldsymbol{Z}_N)=0.
\end{align}

In this paper, we study the $(N,\ell)$ fully secure matrix multiplication problem. In this setting, security has to be preserved when $\ell\leq N$ servers may collude. In other words, despite having access to the collection of encoded matrices $\tilde{\boldsymbol{A}}_{\mathcal{L}}$ and $\tilde{\boldsymbol{B}}_{\mathcal{L}}$, $\mathcal{L}\subseteq[N],|\mathcal{L}|=\ell$, secrecy has to be maintained. Thus, $\tilde{\boldsymbol{A}}_{\mathcal{L}}$ and $\tilde{\boldsymbol{B}}_{\mathcal{L}}$ do not reveal any information on the private matrices $\boldsymbol{A}$ and $\boldsymbol{B}$. This is expressed information-theoretically by the \emph{security constraint} 
\begin{align}
I(\tilde{\boldsymbol{A}}_{\mathcal{L}},\tilde{\boldsymbol{B}}_{\mathcal{L}};\boldsymbol{A},\boldsymbol{B})=0,\quad\forall\mathcal{L}\subseteq[N],|\mathcal{L}|=\ell.
\end{align}
Next, we define two conflicting metrics -- $(i)$ \emph{rate} and $(ii)$ \emph{recovery threshold} -- which we seek to optimize in subsequent sections.
 
First, we say the rate $R_{N,\ell}$ is \emph{achievable} if there exists $\boldsymbol{f},\boldsymbol{g}$ and $d(\cdot)$ satisfying the decodability and security constraints. The rate $R_{N,\ell}$ is the ratio between the number of desired bits vs. the number of downloaded bits and is thus given by 
\begin{align}
R_{N,\ell}=\frac{H(\boldsymbol{A}\boldsymbol{B})}{\sum_{i=1}^{N}H(\boldsymbol{Z}_i)}.
\end{align} The \emph{capacity} $C$ is the supremum of $R_{N,\ell}$ over all achievable schemes.

Second, we call a secure matrix multiplication strategy to be \emph{$\omega_{N,\ell}$-securely recoverable} if the user can recover the matrix product $\bm A\bm B$ from results of $\Omega\subseteq[N],|\Omega|=\omega_{N,\ell}$ servers while complying with the security constraint when any combination of $\ell\leq N$ servers collude. The \emph{recovery threshold} is the \emph{minimum} integer $\omega_{N,\ell}$ such that the multiplication scheme composed of encoders $\boldsymbol{f},\boldsymbol{g}$ and decoder $d(\cdot)$ is $\omega_{N,\ell}$-securely recoverable.

\section{Aligned secret sharing scheme with Matrix Partition}

In an $(N,\ell)$ fully secure matrix multiplication problem, a user is interested in computing $\bm A\bm B$ using $N$ servers without revealing $\ell$ colluding servers information about $\bm A$ and $\bm B$. To this end, the user breaks $\bm A$ vertically into $r_{\bm A}$ sub-matrices and $\bm B$ horizontally into $r_{\bm B}$ sub-matrices, i.e.,
\begin{align*}
\bm A= \begin{bmatrix}
\bm A_1 \\
\bm A_2 \\
\vdots \\
\bm A_{r_{\bm A}}
\end{bmatrix} \text{ and } \bm B= \begin{bmatrix}
\bm B_1 & \bm B_2 & \hdots & \bm B_{r_{\bm B}}
\end{bmatrix}.
\end{align*} Thus, we get $\bm A$ and $\bm B$ by concatenating the sub-matrices $\bm A_i \in \mathbb{F}^{(m/r_{\bm A})\times n},\; i \in [r_{\bm A}]$ and $\bm B_j \in \mathbb{F}^{n\times (p/r_{\bm B})},\; j\in[r_{\bm B}]$. \textcolor{black}{The number of rows $m$ and $n$ are multiple of $r_{\bm A}$ and $r_{\bm B}$, respectively. Under the proposed matrix partition, the matrix product is given by} 
\begin{align*}
\begin{bmatrix}
\bm A_1 \bm B_1 & \bm A_1 \bm B_2 & \hdots & \bm A_1 \bm B_{r_{\bm B}} \\
\bm A_2 \bm B_1 & \bm A_2 \bm B_2 & \hdots & \bm A_2 \bm B_{r_{\bm B}} \\
\vdots & \vdots & \ddots & \vdots \\
\bm A_{r_{\bm A}} \bm B_1 & \bm A_{r_{\bm A}} \bm B_2 & \hdots & \bm A_{r_{\bm A}} \bm B_{r_{\bm B}}
\end{bmatrix}.
\end{align*} The user encodes the matrices $\bm A$ and $\bm B$ according to 
\begin{align*}
\tilde{\bm A}_i & = \sum\limits_{j=1}^{r_{\bm A}} \bm A_j x_i ^{(j-1)} + \sum\limits_{k=1}^{\ell} \bm K_{\bm A_k} x_{i}^{(k+r_{\bm A}-1)},\\
\tilde{\bm B}_i  & = \sum\limits_{j=1}^{r_{\bm B}} \bm B_j x_i ^{\left(j-1\right)\left(r_{\bm A}+\ell\right)} + \sum\limits_{k=1}^{\ell} \bm K_{\bm B_k} x_{i}^{\left(k+r_{\bm A}-1\right)+\left(r_{\bm B}-1\right)\left(r_{\bm A}+\ell\right)},
\end{align*} where all entries of matrices $\bm K_{\bm A_1},\ldots, \bm K_{\bm A_\ell}\in \mathbb{F}^{(m/r_{\bm A})\times n}$ and $\bm K_{\bm B_1},\ldots, \bm K_{\bm B_\ell}\in \mathbb{F}^{n\times (p/r_{\bm B})}$ are i.i.d. uniform random variables. The exponents of the $x_i$ are carefully chosen to facilitate the alignment of undesired components \cite{Kakar17}. Details are discussed in the next paragraph. The user sends the pair $\big(\tilde{\bm A}_{i},\tilde{\bm B}_{i}\big)$ to the server where server $i$ in return 
computes $\bm Z_i=\tilde{\bm A}_{i}\tilde{\bm B}_{i}$ and sends its \emph{answer} $\bm Z_i$ back to the user. The user seeks to retrieve $\bm A\bm B$ by observing up to $N$ polynomials $p(x_i)$, $i=1,\ldots,N$ of degree $Q_{N,\ell}-1$, where $Q_{N,\ell}\triangleq(r_{\bm A}+\ell)(r_{\bm B}+1)-1$\footnote{We frequently omit using the first or the second subscript when $N$ or $\ell$ remain constant, e.g., we simply write $Q$ to denote $Q_{N,\ell}$ for constant $(N,\ell)$. In almost all cases where $N$ is of no concern, we omit the first index and write $Q_{\ell}$.}. The polynomial corresponds to $p(x)=\sum_{i=1}^{4}p_{i}(x)$ and is given by
\begin{align}
p\left( x\right)&=\underbrace{\sum\limits_{j=1}^{r_{\bm A}}\sum\limits_{j'=1}^{\alert{r_{\bm B}}} \bm A_j \bm B_{j'} x^{j+\left(j'-1\right)\left(r_{\bm A}+\ell\right)-1}}_{\triangleq \textcolor{blue}{p_1(x)}}+\underbrace{\sum\limits_{k=1}^{\ell}\sum\limits_{j'=1}^{r_{\bm B}} \bm K_{\bm A_k} \bm B_{j'} x^{k+r_{\bm A}+\left(j'-1\right)\left(r_{\bm A}+\ell\right)-1}}_{\triangleq \textcolor{magenta}{p_2(x)}}\nonumber\\&+\underbrace{\sum\limits_{j=1}^{r_{\bm A}}\sum\limits_{k'=1}^{\ell} \bm A_j \bm K_{\bm B_{k'}} x^{j+k'+\left(r_{\bm A}-1\right)+\left(r_{\bm B}-1\right)\left(r_{\bm A}+\ell\right)-1}}_{\triangleq \textcolor{olive}{p_3(x)}}+\underbrace{\sum\limits_{k=1}^{\ell}\sum\limits_{k'=1}^{\ell} \bm K_{\bm A_k} \bm K_{\bm B_{k'}} x^{k+k^{'}+2\left(r_{\bm A}-1\right)+\left(r_{\bm B}-1\right)\left(r_{\bm A}+\ell\right)}}_{\triangleq \textcolor{teal}{p_4(x)}}.
\end{align}

\begin{figure*}
	\centering  
	\begin{tikzpicture}[decoration=brace]
	\draw(0,0)--(12.35,0);
	\foreach \x/\xtext in {0/{\scriptsize 0},0.65/{\scriptsize 1},1.3/{\scriptsize 2},2.6/{\scriptsize $r_{\bm A}-1$},3.25/{\scriptsize $r_{\bm A}$},6.5/{\scriptsize $r_{\bm A}+\ell$},9.75/{\scriptsize $2(r_{\bm A}+\ell)-\ell$},12.35/{\scriptsize $2(r_{\bm A}+\ell)-1$}} 
	\draw(\x,5pt)--(\x,-5pt) node[below] {\xtext};
	\draw[blue,fill=blue] (0,0) circle (0.5ex);
	\draw[blue,fill=blue] (0.65,0) circle (.5ex);
	\draw[blue,fill=blue] (1.3,0) circle (.5ex);
	\draw[blue,fill=blue] (1.95,0) circle (.5ex);
	\draw[blue,fill=blue] (2.6,0) circle (.5ex);
	\draw[magenta,fill=magenta] (3.25,0) circle (.5ex);  
	\draw[magenta,fill=magenta] (3.9,0) circle (.5ex);
	\draw[magenta,fill=magenta] (4.55,0) circle (.5ex);
	\draw[magenta,fill=magenta] (5.2,0) circle (.5ex);
	\draw[magenta,fill=magenta] (5.85,0) circle (.5ex);  
	\draw[blue,fill=blue] (6.5,0) circle (.5ex);
	\draw[blue,fill=blue] (7.15,0) circle (.5ex);
	\draw[blue,fill=blue] (7.8,0) circle (.5ex);
	\draw[blue,fill=blue] (8.45,0) circle (.5ex);
	\draw[blue,fill=blue] (9.1,0) circle (.5ex);
	\draw[magenta,fill=magenta] (9.75,0) circle (.5ex);
	\draw[magenta,fill=magenta] (10.4,0) circle (.5ex);  
	\draw[magenta,fill=magenta] (11.05,0) circle (.5ex);
	\draw[magenta,fill=magenta] (11.7,0) circle (.5ex);
	\draw[magenta,fill=magenta] (12.35,0) circle (.5ex); 
	\draw [<-,yshift=1.5ex] (2.6,0) -- (1.3,0) node[above=0.0ex] {{\scriptsize \textcolor{blue}{$r_{\bm A}$}}};
	\draw [->,yshift=1.5ex] (1.3,0) -- (0,0);
	\draw[decorate, yshift=-4ex] (2.6,0) -- node[below=0.4ex] {{\footnotesize \textcolor{blue}{$p_1(x)$}}} (0,0);
	\draw[decorate, yshift=-4ex] (5.85,0) -- node[below=0.4ex] {{\footnotesize\textcolor{magenta}{$p_2(x)$}}} (3.25,0);
	\draw [<-,yshift=1.5ex] (9.1,0) -- (7.8,0) node[above=0.0ex] {{\scriptsize \textcolor{blue}{$r_{\bm A}$}}};
	\draw [->,yshift=1.5ex] (7.8,0) -- (6.5,0);
	\draw[decorate, yshift=-4ex] (9.1,0) -- node[below=0.4ex] {{\footnotesize \textcolor{blue}{$p_1(x)$}}} (6.5,0);    
	\draw[decorate, yshift=-4ex] (12.35,0) -- node[below=0.4ex] {{\footnotesize\textcolor{magenta}{$p_2(x)$}}} (9.75,0);
	
	\foreach \x in {0.65,1.3,...,3.25}  \draw (\x,0) -- (\x,-3pt);
	\foreach \x in {3.25,3.9,...,6.5} \draw (\x,0) -- (\x,-3pt);
	\foreach \x in {6.5,7.15,...,9.75} \draw (\x,0) -- (\x,-3pt);
	\foreach \x in {9.75,10.4,...,12.35} \draw (\x,0) -- (\x,-3pt);

	\node at (1.65,-0.2) {\ldots};
	\node at (4.9,-0.2) {\ldots};
	\node at (8.15,-0.2) {\ldots};
	\node at (11.4,-0.2) {\ldots};
	\node at (0,-1.25) {\vdots};

	\begin{scope}[shift={(0,-2.2)}]
	\draw(0,0)--(12.35,0);
	\foreach \x/\xtext in {0/{\scriptsize  $(r_{B}-3)(r_{\bm A}+\ell)$},3.25/{\scriptsize $(r_{B}-2)(r_{\bm A}+\ell ) -\ell$},6.5/{\scriptsize $(r_{B}-2)(r_{\bm A}+\ell ) $},9.75/{\scriptsize $(r_{B}-1)(r_{\bm A}+\ell)-\ell$ },12.35/{\scriptsize $(r_{B}-1)(r_{\bm A}+\ell)-1$}}
	\draw(\x,5pt)--(\x,-5pt) node[below] {\xtext};
	\draw[blue,fill=blue] (0,0) circle (.5ex);
	\draw[blue,fill=blue] (0.65,0) circle (.5ex);
	\draw[blue,fill=blue] (1.3,0) circle (.5ex);
	\draw[blue,fill=blue] (1.95,0) circle (.5ex);
	\draw[blue,fill=blue] (2.6,0) circle (.5ex);
	\draw[magenta,fill=magenta] (3.25,0) circle (.5ex);  
	\draw[magenta,fill=magenta] (3.9,0) circle (.5ex);
	\draw[magenta,fill=magenta] (4.55,0) circle (.5ex);
	\draw[magenta,fill=magenta] (5.2,0) circle (.5ex);
	\draw[magenta,fill=magenta] (5.85,0) circle (.5ex);  
	\draw[blue,fill=blue] (6.5,0) circle (.5ex);
	\draw[blue,fill=blue] (7.15,0) circle (.5ex);
	\draw[blue,fill=blue] (7.8,0) circle (.5ex);
	\draw[blue,fill=blue] (8.45,0) circle (.5ex);
	\draw[blue,fill=blue] (9.1,0) circle (.5ex);
	\draw[magenta,fill=magenta] (9.75,0) circle (.5ex);
	\draw[magenta,fill=magenta] (10.4,0) circle (.5ex);  
	\draw[magenta,fill=magenta] (11.05,0) circle (.5ex);
	\draw[magenta,fill=magenta] (11.7,0) circle (.5ex);
	\draw[magenta,fill=magenta] (12.35,0) circle (.5ex);
	\draw [<-,yshift=1.5ex] (2.6,0) -- (1.3,0) node[above=0.0ex] {{\scriptsize \textcolor{blue}{$r_{\bm A}$}}};
	\draw [->,yshift=1.5ex] (1.3,0) -- (0,0);
	\draw[decorate, yshift=-4ex] (2.6,0) -- node[below=0.4ex] {{\footnotesize \textcolor{blue}{$p_1(x)$}}} (0,0);
	\draw [<-,yshift=1.5ex] (9.1,0) -- (7.8,0) node[above=0.0ex] {{\scriptsize 		    \textcolor{blue}{$r_{\bm A}$}}};
	\draw [->,yshift=1.5ex] (7.8,0) -- (6.5,0);
	\draw[decorate, yshift=-4ex] (5.85,0) -- node[below=0.4ex] {{\footnotesize \textcolor{magenta}{$p_2(x)$}}} (3.25,0);
	\draw[decorate, yshift=-4ex] (9.1,0) -- node[below=0.4ex] {{\footnotesize \textcolor{blue}{$p_1(x)$}}} (6.5,0);    
	\draw[decorate, yshift=-4ex] (12.35,0) -- node[below=0.4ex] {{\footnotesize\textcolor{magenta}{$p_2(x)$}}} (9.75,0);
	
	\foreach \x in {0.65,1.3,...,3.25}  \draw (\x,0) -- (\x,-3pt);
	\foreach \x in {3.25,3.9,...,6.5} \draw (\x,0) -- (\x,-3pt);
	\foreach \x in {6.5,7.15,...,9.75} \draw (\x,0) -- (\x,-3pt);
	\foreach \x in {9.75,10.4,...,12.35} \draw (\x,0) -- (\x,-3pt);

	\node at (1.65,-0.2) {\ldots};
	\node at (4.9,-0.2) {\ldots};
	\node at (8.15,-0.2) {\ldots};
	\node at (11.4,-0.2) {\ldots};
	\end{scope}

	\begin{scope}[shift={(0,-4.4)}]
	\draw(0,0)--(12.35,0);
	\foreach \x/\xtext in {0/{\scriptsize  $(r_{B}-1)(r_{\bm A}+\ell)$},3.25/{\scriptsize  $r_{B}(r_{\bm A}+\ell ) -\ell$},6.5/{\scriptsize $r_{B}(r_{\bm A}+\ell ) $},9.75/{\scriptsize $(r_{\bm B}+1)(r_{\bm A} +\ell)-\ell$ },12.35/{\scriptsize $(r_{\bm B}+1)(r_{\bm A} +\ell)-2$}}
	\draw(\x,5pt)--(\x,-5pt) node[below] {\xtext};
	\draw[blue,fill=blue] (0,0) circle (.5ex);
	\draw[blue,fill=blue] (0.65,0) circle (.5ex);
	\draw[blue,fill=blue] (1.3,0) circle (.5ex);
	\draw[blue,fill=blue] (1.95,0) circle (.5ex);
	\draw[blue,fill=blue] (2.6,0) circle (.5ex);
	\draw[magenta,fill=magenta] (3.25,0) circle (.5ex);  
	\draw[olive,fill=olive] (3.25,1ex) circle (.5ex);
	\draw[teal,fill=teal] (3.25,2ex) circle (.5ex);    
	\draw[magenta,fill=magenta] (3.9,0) circle (.5ex);
	\draw[olive,fill=olive] (3.9,1ex) circle (.5ex);
	\draw[teal,fill=teal] (3.9,2ex) circle (.5ex);
	\draw[magenta,fill=magenta] (4.55,0) circle (.5ex);
	\draw[olive,fill=olive] (4.55,1ex) circle (.5ex);
	\draw[teal,fill=teal] (4.55,2ex) circle (.5ex);
	\draw[magenta,fill=magenta] (5.2,0) circle (.5ex);
	\draw[olive,fill=olive] (5.2,1ex) circle (.5ex);
	\draw[teal,fill=teal] (5.2,2ex) circle (.5ex);
	\draw[magenta,fill=magenta] (5.85,0) circle (.5ex); 
	\draw[olive,fill=olive] (5.85,1ex) circle (.5ex);
	\draw[teal,fill=teal] (5.85,2ex) circle (.5ex); 
	\draw[olive,fill=olive] (6.5,0) circle (.5ex);
	\draw[teal,fill=teal] (6.5,1ex) circle (.5ex);
	\draw[olive,fill=olive] (7.15,0) circle (.5ex);
	\draw[teal,fill=teal] (7.15,1ex) circle (.5ex);
	\draw[olive,fill=olive] (7.8,0) circle (.5ex);
	\draw[teal,fill=teal] (7.8,1ex) circle (.5ex);
	\draw[olive,fill=olive] (8.45,0) circle (.5ex);
	\draw[teal,fill=teal] (8.45,1ex) circle (.5ex);
	\draw[teal,fill=teal] (9.1,0) circle (.5ex);
	\draw[teal,fill=teal] (9.75,0) circle (.5ex);
	\draw[teal,fill=teal] (10.4,0) circle (.5ex);  
	\draw[teal,fill=teal] (11.05,0) circle (.5ex);
	\draw[teal,fill=teal] (11.7,0) circle (.5ex);
	\draw[teal,fill=teal] (12.35,0) circle (.5ex);
	\draw [<-,yshift=1.5ex] (2.6,0) -- (1.3,0) node[above=0.0ex] {{\scriptsize \textcolor{blue}{$r_{\bm A}$}}};
	\draw [->,yshift=1.5ex] (1.3,0) -- (0,0);
	\draw[decorate, yshift=-4ex] (2.6,0) -- node[below=0.4ex] {{\footnotesize \textcolor{blue}{$p_1(x)$}}} (0,0);
	\draw[decorate, yshift=-4ex] (5.85,0) -- node[below=0.4ex] {{\footnotesize$\textcolor{magenta}{p_2(x)}, \; \textcolor{olive}{p_3(x)}, \; \textcolor{teal}{p_4(x)}$}} (3.25,0);
	\draw[decorate, yshift=-4ex] (8.45,0) -- node[below=0.4ex] {{\footnotesize$\textcolor{olive}{p_3(x)}, \; \textcolor{teal}{p_4(x)}$}} (6.5,0);    
	\draw[decorate, yshift=-4ex] (12.35,0) -- node[below=0.4ex] {{\footnotesize\textcolor{teal}{$p_4(x)$}}} (9.1,0);
	\draw[decorate, yshift=-8ex] (12.35,0) -- node[below=0.4ex] {{ \small{Alignment}}} (3.25,0);
	\foreach \x in {0.65,1.3,...,3.25}  \draw (\x,0) -- (\x,-3pt);
	\foreach \x in {3.25,3.9,...,6.5} \draw (\x,0) -- (\x,-3pt);
	\foreach \x in {6.5,7.15,...,9.75} \draw (\x,0) -- (\x,-3pt);
	\foreach \x in {9.75,10.4,...,12.35} \draw (\x,0) -- (\x,-3pt);

	\node at (1.65,-0.2) {\ldots};
	\node at (4.9,-0.2) {\ldots};
	\node at (8.15,-0.2) {\ldots};
	\node at (11.4,-0.2) {\ldots};
	\end{scope}
	\end{tikzpicture}
	\caption{\small Number line of the exponent of the polynomial $p(x)$ and its association to the terms $p_{i}(x),i=1,\ldots,4$.}
	\label{fig:num_line}  
\end{figure*}
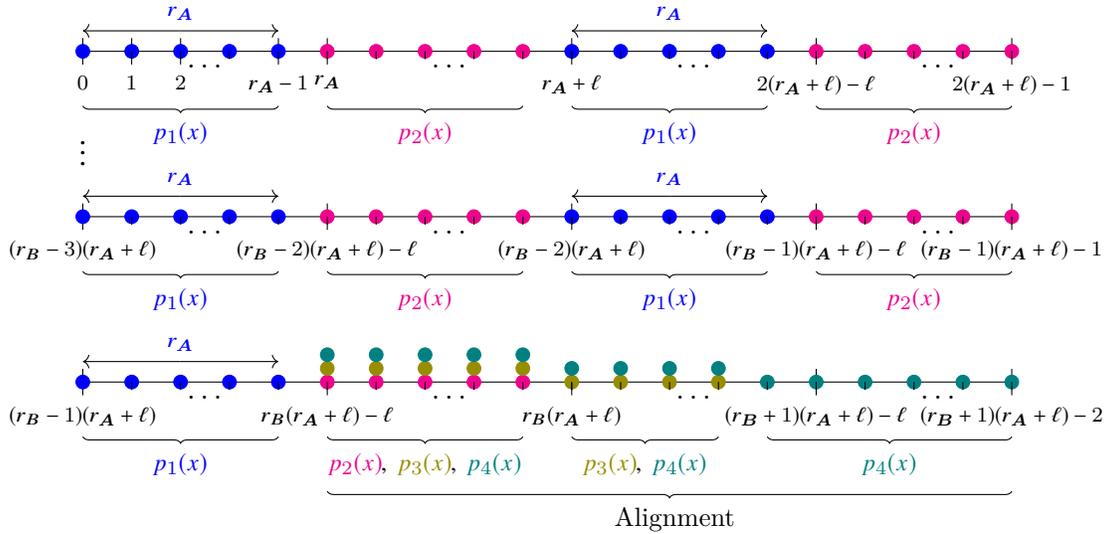 

To reconstruct $\bm A\bm B$, the user is interested in $p_1(x)$. The remaining terms $p_{i}(x),i=2,\ldots,4,$ can be thought of \emph{interference}. Thus, with respect to $p_1(x)$ each exponent in $x$ needs to have only one attributable item $\bm A_j\bm B_{j'}$ to distinguish desired components from each other and also from undesired components $\bm K_{\bm A_k} \bm B_{j'}$, $\bm A_j \bm K_{\bm B_{k'}}$ and $\bm K_{\bm A_k} \bm K_{\bm B_{k'}}$. One can verify that each exponent of $p_1(x)$ does not occur in the remaining undesired terms $p_{i}(x),i=2,\ldots,4$. In contrast, there are multiple items assigned to the remaining exponents not being included in $p_{1}(x)$. In other words, we \emph{align} multiple undesired items to single exponents. Thus, this scheme is called an \emph{aligned secret sharing} scheme. A pictorial representation of the association of components to exponents is provided in Fig. \ref{fig:num_line}. The exponents of $x_i$ in $\tilde{\bm A}_i$ and $\tilde{\bm B}_i$ are chosen in the encoding process to avoid overlaps of desired terms $(i)$ with each other and $(ii)$ with undesired terms while simultaneously create as many alignment opportunities as possible when computing $\bm Z_i=\tilde{\bm A}_i\tilde{\bm B}_i$. The desired terms consume $r_{\bm A}r_{\bm B}$ exponents while the interference occupies $\ell(r_{\bm B}+1)+r_{\bm A}-1$ exponents. More specifically, the first $\ell(r_{\bm B}-1)$ components of $p_{2}(x)$ are not aligned with other interference components of $p_{3}(x)$ and $p_{4}(x)$. In contrast, the remaining $r_{\bm A}+2\ell-1$ exponents of $p_2(x)$, $p_{3}(x)$ or $p_{4}(x)$ are subject to (subspace) alignment of at least two components.  

Recall that the polynomial $p(x)$ has a degree of $Q-1$ and the user has access to $N$ observations. In order to enable decoding, we have to ensure that the degree of the polynomial does not exceed the total number of available servers, or observations, $N$, i.e., 
\begin{align}
\label{PIneq}
\alert{Q_{N,\ell}(r_{\bm A},r_{\bm B})}\triangleq\left(r_{\bm A}+\ell\right)\left(r_{\bm B}+1\right)-1\leq N.
\end{align} Thus, the user can retrieve its desired items by using polynomial interpolation. Since the user recovers $r_{\bm A}r_{\bm B}$ desired items out of $\alert{Q_{N,\ell}}(r_{\bm A},r_{\bm B})$ calculated items, the aligned secret sharing scheme achieves a rate of 
\begin{align*}
\alert{R_{N,\ell}(r_{\bm A},r_{\bm B})\triangleq}\frac{r_{\bm A}r_{\bm B}}{\alert{Q_{N,\ell}(r_{\bm A},r_{\bm B})}}=\frac{r_{\bm A}r_{\bm B}}{\left(r_{\bm A}+\ell\right)\left(r_{\bm B}+1\right)-1}.
\end{align*}
In order to maximize the rate \alert{$R_{N,\ell}(r_{\bm A},r_{\bm B})$\footnote{For the sake of simplicity, the notation of $R_{N,\ell}(r_{\bm A},r_{\bm B})$ is aligned with that of $Q_{N,\ell}(r_{\bm A},r_{\bm B})$.}}, we need to solve the optimization problem
\begin{subequations}\label{opt}
	\begin{alignat}{2}
		\max_{r_{\bm A},r_{\bm B}} \qquad & \frac{r_{\bm A}r_{\bm B}}{\left(r_{\bm A}+\ell\right)\left(r_{\bm B}+1\right)-1} \tag{\ref{opt}} \\
		\text{subject to} \qquad & \left(r_{\bm A}+\ell\right)\left(r_{\bm B}+1\right)-1\leq N\\
		&  r_{\bm A}, r_{\bm B} \in \alert{\mathbb{Z}^{+}}.
	\end{alignat}
\end{subequations} We denote the optimal decision variables and objective value of \eqref{opt} by $(r^{\star}_{\bm A},r^{\star}_{\bm B})$ and $R^{\star}$.  

Further, we note that the effective number of server observations the user needs to determine the matrix product $\bm A\bm B$ is $Q$. Thus, the aligned secret sharing strategy is \emph{$Q$-securely recoverable}. In order to make the aligned secret sharing scheme less prone to slower computing servers, or \emph{stragglers}, one has to solve the optimization problem 
\begin{subequations}\label{opt_recover_th}
	\begin{alignat}{3}
	\min_{r_{\bm A},r_{\bm B}} \qquad & \left(r_{\bm A}+\ell\right)\left(r_{\bm B}+1\right)-1 \tag{\ref{opt_recover_th}} \\
	\text{subject to} \qquad & \frac{r_{\bm A}r_{\bm B}}{\left(r_{\bm A}+\ell\right)\left(r_{\bm B}+1\right)-1}\geq R_{\text{th}} \\
	& \left(r_{\bm A}+\ell\right)\left(r_{\bm B}+1\right)-1\leq N\\
	& r_{\bm A}, r_{\bm B} \in \alert{\mathbb{Z}^{+}}.
	\end{alignat}
\end{subequations}

The optimization problems \eqref{opt} and \eqref{opt_recover_th} find the best choice of how to partition the left and right matrices at the user for given $N$ and $\ell$ (and a minimum rate requirement $R_{\text{th}}\leq R^{\star}$ in \eqref{opt_recover_th}). For both problems, we propose close-to-optimal analytical solutions. The solutions to the rate maximization problem \eqref{opt} and $Q$-secure recoverability problem are stated in Theorem \ref{th:res} and \ref{th:res_recover_th}, respectively. To differentiate the (optimal) solution of \eqref{opt_recover_th} from \eqref{opt}, we use the breve mark $(\:\breve{}\:)$ instead of the star symbol $(^{\star})$, e.g., $\breve{Q}$ instead of $Q^{\star}$.  

\section{Solution of the Matrix Partitioning Problem \eqref{opt} and Discussion}
\label{sec:main_result_discussion} 

\begin{theorem}
	\label{th:res}
	The solution $(\hat{r}_{\bm A},\hat{r}_{\bm B})$ is a close-to-optimal analytical solution to the optimization problem \eqref{opt} for given parameters $N$ and $\ell$. Hereby, 
	\begin{align}\label{eq:est_r_b}
	\hat{r}_{\bm B}=\max\Bigg\{1,\:\biggl\lceil-\frac{3}{2}+\sqrt{\frac{1}{4}+\frac{N}{\ell}}\:\biggr\rceil\:\Bigg\}
	\end{align} 
	and $\hat{r}_{\bm A}$ is the largest possible integer $r_{\bm A}\geq 1$ that satisfies the inequality
	\begin{align}\label{eq:est_r_a}
	(r_{\bm A}+\ell)(\hat{r}_{\bm B}+1)-1\leq N.
	\end{align}  
\end{theorem}
\begin{proof}
	The proof is based on the inductive approach of deriving the relationship between consecutive optimal solution pairs $(r^{\star}_{\ell-1,\bm A},r^{\star}_{\ell-1,\bm B})$ and $(r^{\star}_{\ell,\bm A},r^{\star}_{\ell,\bm B})$. Ultimately, under some additional approximations, this helps us in deriving \eqref{eq:est_r_b} and \eqref{eq:est_r_a}. For further details, we refer the reader to Section \ref{sec:close_opt_sol}.
\end{proof}

\begin{remark}[Upper bound]
	The best information-theoretic upper bound known of the two-sided matrix multiplication problem on the rate is derived in \cite{Chang2018OnTC}. The upper bound of the two-sided model is in fact the one-sided model for which the capacity is known to be $C_{\text{one-sided}}=\frac{N-\ell}{N}$.  
\end{remark}

\begin{remark} 
	Compared to the scheme proposed by Chang and Tandon (CT) \cite[Theorem 2]{Chang2018OnTC}, our aligned secret sharing scheme significantly improves on the communication rate (see Fig. \ref{plot:Rate_over_ell}). This is illustrated in Fig. \ref{plot:Rate_over_ell} when comparing the achievable communication rate of 'Unequal' and 'Equal' with 'CT'. Specifically, while our scheme ensures a non-zero rate for at most $\lfloor\nicefrac{(N-1)}{2}\rfloor$ colluding servers, CTs scheme support only $\lfloor\sqrt{N}-1\rfloor$ colluding servers. Further, appropriate matrix partition is of importance when comparing the achievable rates of optimized (or unequal) and equal $(r_{\bm A}=r_{\bm B})$ partitions in Fig. \ref{plot:Rate_over_ell}. The unequal partitions use the partitioning proposed in Theorem \ref{th:res}.     
\end{remark}

\begin{figure}[h]
	\centering
	\input{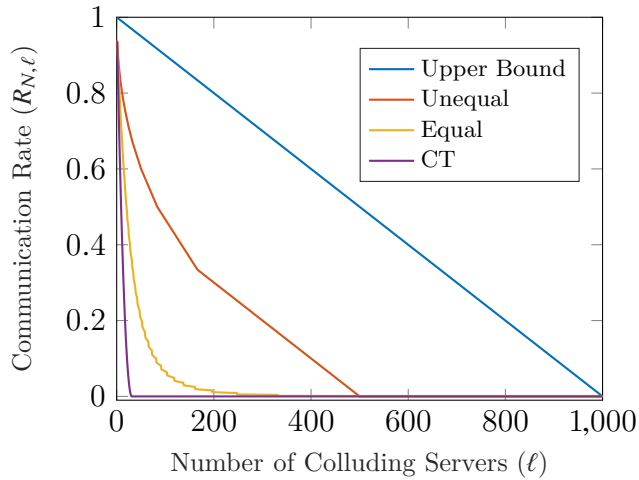}
	\caption{\footnotesize Comparison between the achievable communication rates for (i) one-sided secure matrix multiplication (which is an upper bound on two-sided multiplication), (ii) unequally and (iii) equally partitioned aligned secret sharing scheme and (iv) the scheme proposed by Chang and Tandon for $N=1000$ as a function of the number of colluding servers $\ell$.}
	\label{plot:Rate_over_ell}
\end{figure}


\begin{figure}[h]
	\centering
	\input{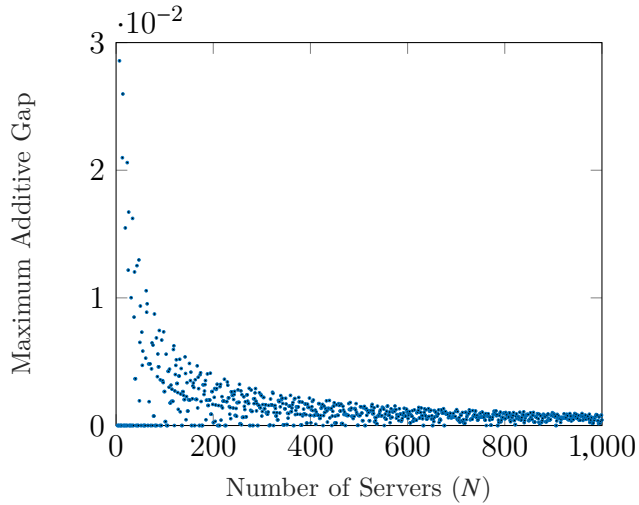}
	\caption{\footnotesize Plot of the maximum additive gap $\max_{\ell\in[N]}|R^{\star}_{\ell}-\hat{R}_{\ell}|$ as a function of $N$. }
	\label{plot:Error_over_N}
\end{figure}
\begin{figure}[h]
	\centering
%
%
\definecolor{mycolor1}{rgb}{0.00000,0.44700,0.74100}%
\begin{tikzpicture}

\begin{axis}[%
width=0.75*4.521*0.75in,
height=0.75*3.566*0.75in,
at={(0.758in,0.481in)},
scale only axis,
xmin=0,
xmax=1000,
xlabel style={font=\color{white!15!black}},
xlabel={\small{Number of Servers ($N$)}},
ymin=0,
ymax=8.2,
ylabel style={font=\color{white!15!black}},
ylabel={\small{$\#$ of Sub-Optimal Solutions}},
axis background/.style={fill=white},
legend style={legend cell align=left, align=left, draw=white!15!black}
]
\addplot[only marks, mark=*, mark options={}, mark size=0.5000pt, draw=mycolor1] table[row sep=crcr]{%
x	y\\
3	0\\
4	0\\
5	0\\
6	0\\
7	1\\
8	0\\
9	0\\
10	0\\
11	0\\
12	0\\
13	1\\
14	1\\
15	0\\
16	0\\
17	0\\
18	0\\
19	2\\
20	0\\
21	0\\
22	0\\
23	1\\
24	0\\
25	1\\
26	1\\
27	0\\
28	0\\
29	0\\
30	0\\
31	2\\
32	0\\
33	0\\
34	2\\
35	0\\
36	0\\
37	1\\
38	1\\
39	1\\
40	1\\
41	0\\
42	0\\
43	3\\
44	0\\
45	0\\
46	0\\
47	2\\
48	1\\
49	1\\
50	1\\
51	0\\
52	0\\
53	2\\
54	2\\
55	2\\
56	0\\
57	0\\
58	0\\
59	0\\
60	0\\
61	1\\
62	2\\
63	1\\
64	1\\
65	0\\
66	0\\
67	3\\
68	1\\
69	2\\
70	2\\
71	0\\
72	0\\
73	1\\
74	2\\
75	0\\
76	2\\
77	2\\
78	2\\
79	3\\
80	0\\
81	0\\
82	0\\
83	2\\
84	0\\
85	1\\
86	1\\
87	3\\
88	3\\
89	1\\
90	0\\
91	2\\
92	0\\
93	0\\
94	2\\
95	2\\
96	2\\
97	2\\
98	2\\
99	1\\
100	1\\
101	1\\
102	1\\
103	4\\
104	0\\
105	0\\
106	0\\
107	3\\
108	3\\
109	1\\
110	1\\
111	1\\
112	1\\
113	1\\
114	1\\
115	2\\
116	2\\
117	2\\
118	3\\
119	1\\
120	0\\
121	1\\
122	1\\
123	1\\
124	1\\
125	1\\
126	1\\
127	3\\
128	1\\
129	2\\
130	2\\
131	1\\
132	1\\
133	2\\
134	3\\
135	0\\
136	0\\
137	2\\
138	2\\
139	3\\
140	1\\
141	1\\
142	1\\
143	3\\
144	1\\
145	2\\
146	1\\
147	1\\
148	1\\
149	1\\
150	1\\
151	4\\
152	1\\
153	2\\
154	4\\
155	0\\
156	0\\
157	1\\
158	1\\
159	3\\
160	3\\
161	0\\
162	0\\
163	3\\
164	1\\
165	1\\
166	1\\
167	4\\
168	3\\
169	2\\
170	1\\
171	0\\
172	0\\
173	2\\
174	2\\
175	2\\
176	0\\
177	0\\
178	0\\
179	2\\
180	2\\
181	2\\
182	2\\
183	1\\
184	1\\
185	1\\
186	2\\
187	5\\
188	3\\
189	2\\
190	2\\
191	1\\
192	1\\
193	2\\
194	3\\
195	1\\
196	1\\
197	3\\
198	3\\
199	3\\
200	1\\
201	1\\
202	1\\
203	3\\
204	1\\
205	2\\
206	4\\
207	4\\
208	4\\
209	0\\
210	0\\
211	2\\
212	0\\
213	0\\
214	2\\
215	3\\
216	1\\
217	2\\
218	2\\
219	1\\
220	2\\
221	2\\
222	2\\
223	5\\
224	0\\
225	0\\
226	0\\
227	4\\
228	4\\
229	4\\
230	3\\
231	1\\
232	1\\
233	2\\
234	2\\
235	3\\
236	1\\
237	3\\
238	3\\
239	2\\
240	2\\
241	2\\
242	1\\
243	1\\
244	3\\
245	1\\
246	1\\
247	4\\
248	2\\
249	3\\
250	3\\
251	2\\
252	1\\
253	2\\
254	3\\
255	1\\
256	1\\
257	3\\
258	1\\
259	4\\
260	2\\
261	2\\
262	2\\
263	4\\
264	2\\
265	2\\
266	2\\
267	2\\
268	2\\
269	3\\
270	3\\
271	4\\
272	1\\
273	1\\
274	3\\
275	0\\
276	0\\
277	1\\
278	3\\
279	2\\
280	2\\
281	0\\
282	0\\
283	3\\
284	1\\
285	3\\
286	5\\
287	4\\
288	3\\
289	1\\
290	1\\
291	0\\
292	0\\
293	2\\
294	2\\
295	3\\
296	2\\
297	2\\
298	2\\
299	2\\
300	2\\
301	3\\
302	3\\
303	3\\
304	3\\
305	2\\
306	2\\
307	5\\
308	3\\
309	2\\
310	2\\
311	3\\
312	3\\
313	4\\
314	2\\
315	0\\
316	0\\
317	1\\
318	3\\
319	6\\
320	4\\
321	4\\
322	4\\
323	3\\
324	0\\
325	1\\
326	1\\
327	2\\
328	2\\
329	1\\
330	1\\
331	3\\
332	1\\
333	1\\
334	3\\
335	3\\
336	3\\
337	3\\
338	3\\
339	4\\
340	3\\
341	2\\
342	2\\
343	3\\
344	1\\
345	1\\
346	1\\
347	5\\
348	5\\
349	5\\
350	4\\
351	1\\
352	1\\
353	3\\
354	3\\
355	4\\
356	2\\
357	2\\
358	2\\
359	3\\
360	2\\
361	3\\
362	2\\
363	2\\
364	2\\
365	0\\
366	0\\
367	5\\
368	3\\
369	3\\
370	3\\
371	2\\
372	2\\
373	3\\
374	4\\
375	4\\
376	4\\
377	3\\
378	3\\
379	3\\
380	1\\
381	1\\
382	1\\
383	6\\
384	1\\
385	2\\
386	2\\
387	2\\
388	2\\
389	2\\
390	3\\
391	4\\
392	2\\
393	2\\
394	4\\
395	3\\
396	3\\
397	4\\
398	4\\
399	2\\
400	2\\
401	1\\
402	2\\
403	5\\
404	3\\
405	4\\
406	5\\
407	3\\
408	3\\
409	2\\
410	2\\
411	1\\
412	3\\
413	4\\
414	4\\
415	5\\
416	3\\
417	4\\
418	4\\
419	1\\
420	1\\
421	2\\
422	2\\
423	2\\
424	2\\
425	2\\
426	0\\
427	3\\
428	4\\
429	4\\
430	4\\
431	4\\
432	4\\
433	3\\
434	3\\
435	1\\
436	1\\
437	3\\
438	3\\
439	8\\
440	1\\
441	1\\
442	1\\
443	2\\
444	0\\
445	1\\
446	1\\
447	4\\
448	4\\
449	1\\
450	1\\
451	3\\
452	1\\
453	1\\
454	5\\
455	2\\
456	2\\
457	3\\
458	4\\
459	4\\
460	4\\
461	2\\
462	2\\
463	5\\
464	1\\
465	1\\
466	1\\
467	5\\
468	5\\
469	3\\
470	3\\
471	1\\
472	0\\
473	2\\
474	3\\
475	5\\
476	5\\
477	5\\
478	5\\
479	5\\
480	4\\
481	5\\
482	3\\
483	1\\
484	1\\
485	2\\
486	2\\
487	5\\
488	3\\
489	5\\
490	5\\
491	3\\
492	3\\
493	3\\
494	4\\
495	1\\
496	1\\
497	1\\
498	1\\
499	4\\
500	2\\
501	2\\
502	2\\
503	5\\
504	3\\
505	3\\
506	2\\
507	2\\
508	2\\
509	3\\
510	3\\
511	4\\
512	0\\
513	0\\
514	2\\
515	3\\
516	3\\
517	4\\
518	4\\
519	4\\
520	4\\
521	3\\
522	3\\
523	6\\
524	3\\
525	3\\
526	3\\
527	6\\
528	5\\
529	4\\
530	2\\
531	2\\
532	1\\
533	1\\
534	1\\
535	4\\
536	2\\
537	2\\
538	4\\
539	1\\
540	1\\
541	2\\
542	2\\
543	2\\
544	2\\
545	1\\
546	1\\
547	4\\
548	4\\
549	4\\
550	4\\
551	2\\
552	0\\
553	1\\
554	4\\
555	2\\
556	2\\
557	6\\
558	7\\
559	5\\
560	1\\
561	1\\
562	1\\
563	3\\
564	1\\
565	2\\
566	5\\
567	3\\
568	3\\
569	4\\
570	4\\
571	6\\
572	4\\
573	5\\
574	7\\
575	2\\
576	2\\
577	2\\
578	2\\
579	2\\
580	4\\
581	2\\
582	2\\
583	5\\
584	2\\
585	2\\
586	2\\
587	2\\
588	3\\
589	4\\
590	4\\
591	5\\
592	5\\
593	5\\
594	2\\
595	3\\
596	1\\
597	3\\
598	3\\
599	6\\
600	6\\
601	4\\
602	2\\
603	2\\
604	1\\
605	1\\
606	1\\
607	7\\
608	3\\
609	4\\
610	6\\
611	2\\
612	2\\
613	3\\
614	4\\
615	3\\
616	3\\
617	3\\
618	3\\
619	4\\
620	2\\
621	2\\
622	4\\
623	6\\
624	3\\
625	4\\
626	4\\
627	4\\
628	4\\
629	2\\
630	2\\
631	2\\
632	0\\
633	0\\
634	2\\
635	3\\
636	2\\
637	3\\
638	5\\
639	3\\
640	3\\
641	3\\
642	3\\
643	7\\
644	3\\
645	3\\
646	3\\
647	8\\
648	8\\
649	3\\
650	1\\
651	0\\
652	0\\
653	1\\
654	1\\
655	5\\
656	3\\
657	4\\
658	4\\
659	2\\
660	2\\
661	3\\
662	2\\
663	4\\
664	4\\
665	1\\
666	2\\
667	5\\
668	3\\
669	3\\
670	3\\
671	4\\
672	4\\
673	5\\
674	4\\
675	0\\
676	0\\
677	2\\
678	2\\
679	6\\
680	4\\
681	4\\
682	4\\
683	4\\
684	2\\
685	4\\
686	4\\
687	3\\
688	2\\
689	4\\
690	4\\
691	6\\
692	4\\
693	4\\
694	6\\
695	2\\
696	3\\
697	4\\
698	4\\
699	8\\
700	8\\
701	4\\
702	3\\
703	6\\
704	2\\
705	2\\
706	2\\
707	2\\
708	2\\
709	3\\
710	3\\
711	2\\
712	2\\
713	5\\
714	4\\
715	5\\
716	3\\
717	3\\
718	3\\
719	4\\
720	2\\
721	2\\
722	2\\
723	2\\
724	2\\
725	3\\
726	3\\
727	7\\
728	4\\
729	2\\
730	3\\
731	3\\
732	3\\
733	4\\
734	6\\
735	3\\
736	3\\
737	4\\
738	4\\
739	3\\
740	2\\
741	3\\
742	3\\
743	5\\
744	3\\
745	4\\
746	4\\
747	5\\
748	7\\
749	4\\
750	4\\
751	5\\
752	3\\
753	4\\
754	6\\
755	3\\
756	3\\
757	4\\
758	4\\
759	2\\
760	2\\
761	2\\
762	0\\
763	3\\
764	3\\
765	3\\
766	4\\
767	6\\
768	6\\
769	2\\
770	2\\
771	1\\
772	1\\
773	4\\
774	5\\
775	6\\
776	2\\
777	2\\
778	2\\
779	6\\
780	4\\
781	5\\
782	5\\
783	2\\
784	2\\
785	0\\
786	0\\
787	3\\
788	1\\
789	3\\
790	5\\
791	2\\
792	2\\
793	3\\
794	5\\
795	3\\
796	3\\
797	4\\
798	4\\
799	6\\
800	2\\
801	2\\
802	2\\
803	5\\
804	3\\
805	4\\
806	4\\
807	4\\
808	4\\
809	5\\
810	5\\
811	4\\
812	2\\
813	3\\
814	5\\
815	5\\
816	5\\
817	6\\
818	3\\
819	2\\
820	2\\
821	1\\
822	1\\
823	6\\
824	5\\
825	5\\
826	5\\
827	5\\
828	5\\
829	3\\
830	3\\
831	3\\
832	2\\
833	2\\
834	2\\
835	6\\
836	6\\
837	6\\
838	6\\
839	2\\
840	1\\
841	2\\
842	2\\
843	2\\
844	1\\
845	4\\
846	6\\
847	8\\
848	6\\
849	3\\
850	4\\
851	4\\
852	4\\
853	5\\
854	6\\
855	2\\
856	2\\
857	5\\
858	5\\
859	6\\
860	3\\
861	3\\
862	3\\
863	6\\
864	4\\
865	5\\
866	4\\
867	6\\
868	6\\
869	2\\
870	2\\
871	6\\
872	2\\
873	3\\
874	5\\
875	2\\
876	2\\
877	3\\
878	3\\
879	7\\
880	7\\
881	3\\
882	3\\
883	5\\
884	1\\
885	1\\
886	1\\
887	6\\
888	4\\
889	5\\
890	2\\
891	1\\
892	3\\
893	5\\
894	5\\
895	8\\
896	6\\
897	6\\
898	6\\
899	5\\
900	4\\
901	3\\
902	2\\
903	1\\
904	1\\
905	0\\
906	0\\
907	3\\
908	3\\
909	7\\
910	7\\
911	4\\
912	2\\
913	3\\
914	3\\
915	1\\
916	3\\
917	5\\
918	5\\
919	5\\
920	3\\
921	3\\
922	5\\
923	4\\
924	2\\
925	3\\
926	5\\
927	6\\
928	6\\
929	4\\
930	2\\
931	4\\
932	2\\
933	2\\
934	4\\
935	5\\
936	5\\
937	4\\
938	4\\
939	3\\
940	3\\
941	3\\
942	3\\
943	8\\
944	1\\
945	2\\
946	2\\
947	4\\
948	4\\
949	4\\
950	4\\
951	5\\
952	5\\
953	2\\
954	2\\
955	3\\
956	3\\
957	3\\
958	5\\
959	6\\
960	5\\
961	5\\
962	3\\
963	3\\
964	3\\
965	2\\
966	2\\
967	6\\
968	6\\
969	6\\
970	6\\
971	3\\
972	3\\
973	4\\
974	5\\
975	1\\
976	1\\
977	3\\
978	5\\
979	6\\
980	4\\
981	4\\
982	4\\
983	7\\
984	5\\
985	8\\
986	6\\
987	6\\
988	6\\
989	4\\
990	4\\
991	4\\
992	2\\
993	1\\
994	3\\
995	2\\
996	2\\
997	3\\
998	5\\
999	6\\
1000	3\\
};

\end{axis}
\end{tikzpicture}%
	\caption{\footnotesize Plot of the number of sub-optimal solutions of the provided estimation to the optimization problem \eqref{opt} as a function of $N$. Note that the number of sub-optimal solutions for fixed $N$ and variable $\ell$ is upper bounded by $N$.}
	\label{plot:Num_SubOpt_over_N}	
\end{figure}
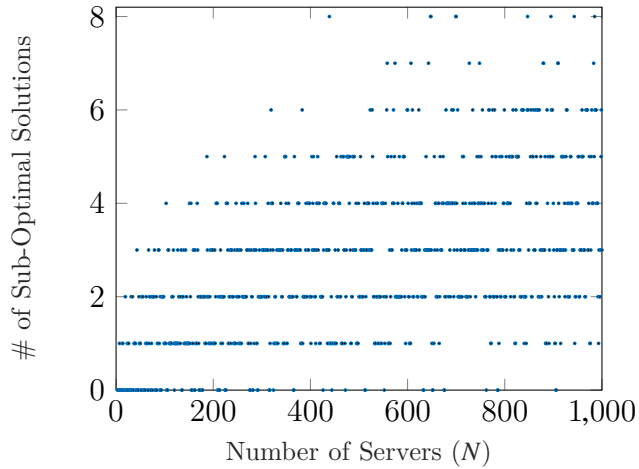

\begin{remark}[Additive gap]
	To evaluate the quality of our proposed analytical solution, we evaluate the maximum additive gap $\max_{\ell\in[N]}|R^{\star}_{\ell}-\hat{R}_{\ell}|$ (see Fig. \ref{plot:Error_over_N}). Hereby, $R^{\star}_{\ell}$ denotes the optimal rate of the optimization problem \eqref{opt} and $\hat{R}_{\ell}$ our proposed estimate. The optimal solution is determined in a brute-force fashion by exhaustive search which is costly in computation. \emph{Numerical} results show that the proposed solution is at most $3\cdot
	10^{-2}$ additively off from the optimal solution.
\end{remark}

\begin{remark} Our proposed solution is frequently the optimal solution. Fig. \ref{plot:Num_SubOpt_over_N} shows the number of sub-optimal solutions of the provided estimation to the optimization problem \eqref{opt} for a given $N$ and variable $\ell$. In theory, for a given $N$, the number of sub-optimal solutions is upper bounded by $N$. Interestingly, this figure shows that '$\#\text{ sub-optimal solutions}\ll N$'. This suggests that our solution solves \eqref{opt} for almost all $\ell\in[N]$ optimally except of very few cases where an almost negligible additive gap is attained.
\end{remark}


\begin{remark}[Server computational complexity]
	We define the \emph{per-server} computational complexity as the number of necessary multiply-accumulate (MAC) operations to determine $\bm Z_{i}= \tilde{\bm A_{i}} \tilde{\bm B_{i}}$. Clearly, the MAC complexity under a general $r_{\bm A}$ and $r_{\bm B}$-partition of matrices $\bm A$ and $\bm B$ becomes $\Theta\left(\frac{mnp}{r_{\bm A}r_{\bm B}}\right)$. Recall that the denominator $r_{\bm A}r_{\bm B}<N$ represents the dimension reserved for desired sub-block matrix products. For constant $N$ and $\ell$, our secret sharing scheme achieves (in comparison to CTs scheme) the better alignment efficiency and thus a larger product $r_{\bm A}r_{\bm B}$. This in return, results in an improved per-server complexity when $m,n$ and $p$ remain constant.  
\end{remark}

\begin{remark}[User decoding complexity]
	The decoding at the user can be interpreted as an interpolation of a $Q-1$-degree polynomial for $\frac{mp}{r_{\bm A}r_{\bm B}}$ times. Hereby, the complexity of a $t$-degree polynomial interpolation is $O(t\log^{2}t\log\log t)$ \cite{Kedlaya11}. Thus, the decoding complexity at the user is of order $O(mp\log^{2}\eta\log\log \eta)$ with $\eta=\max\{r_{\bm A},\ell\}\:r_{\bm B}$.    
\end{remark}

\begin{remark}[Recovery threshold]
	The effective number of server observations the user needs to determine the matrix product $\bm A\bm B$ is (after rate maximization) $Q^{\star}$. In the problem \eqref{opt}, it is desirable to choose $r_{\bm A}$ and $r_{\bm B}$ as large as possible without violating the inequality constraint $Q\leq N$. Typically, after rate maximization we obtain highly straggler-dependent solutions for which $Q^{\star}\approx N$.   
\end{remark}

\begin{remark}[Input matrix dimension]
	Recall that the user splits the input matrices $\bm A$ and $\bm B$ into $r_{\bm A}$ sub-matrices $\bm A_i\in\mathbb{F}^{(m/r_{\bm A})\times n}$ and $r_{\bm B}$ sub-matrices $\bm B_j\in\mathbb{F}^{n\times(p/r_{\bm B})}$. According to Theorem \ref{th:res}, we can easily show from \eqref{eq:est_r_b}
	\begin{align*}
	\max\left\{1,-\frac{3}{2}+\sqrt{\frac{N}{\ell}}\:\right\}\leq r_{\bm B}\leq\sqrt{\frac{N}{\ell}}
	\end{align*} and based on that from \eqref{eq:est_r_a}
	\begin{align*}
	\max\Big\{1,\sqrt{N\ell}-\ell-1\Big\}\leq r_{\bm A}\leq 2\sqrt{N\ell}-\ell+2.
	\end{align*} This suggests that for feasibility in the matrix partitioning, $p$ and $m$ shall (at least) scale according to $\Theta(\sqrt{\nicefrac{N}{\ell}})$ and $\Theta(\sqrt{N\ell})$, respectively.  
\end{remark}

\section{Solution of the Matrix Partitioning Problem \eqref{opt_recover_th}}
\label{sec:main_result_recovery} 

\begin{theorem}
	\label{th:res_recover_th}
	The solution $(\mathring{r}_{\bm A},\mathring{r}_{\bm B})$ is a close-to-optimal analytical solution to the optimization problem \eqref{opt_recover_th} for a given parameter $R_{\text{th}}$ which is feasible with respect to the given parameters $N$ und $\ell$. Hereby, 
	\begin{align}\label{eq:est_r_b1}
	\mathring{r}_{\bm B}=\max\Bigg\{1,\:\biggl\lceil\frac{2}{1-R_{\text{th}}}-2\:\biggr\rceil\:\Bigg\}
	\end{align} 
	and $\mathring{r}_{\bm A}$ is the smallest possible integer $r_{\bm A}\geq 1$ that satisfies the inequality
	\begin{align}\label{eq:est_r_a1}
	\frac{r_{\bm A}\mathring{r}_{\bm B}}{\left(r_{\bm A}+\ell\right)\left(\mathring{r}_{\bm B}+1\right)-1}\geq R_{\text{th}}.
	\end{align}  
\end{theorem}
\begin{proof}
	The proof is very similar to the proof of Theorem \ref{th:res}. The current version of this paper does not include a detailed proof.
\end{proof}

\section{Concluding Remarks}
\label{sec:conclusion}

In this paper, we studied the two-sided secure matrix multiplication problem, where a user is interested in the matrix product $\bm A\bm B$ of two private matrices $\bm A$ and $\bm B$. The user tries to conceal the private matrices from $N$ servers (where we allow for up to $\ell$ servers to collude), but uses them to compute the matrix product. We propose a partition-based aligned secret sharing scheme. Next, we formulate and solve two optimization problems that determine the optimal matrix partition of input matrices $\bm A$ and $\bm B$ to (i) maximize the communication rate of this scheme and (ii) to maximize the recovery threshold. With respect to objective (i), numerical results show that this scheme significantly outperforms the state-of-the-art scheme of Chang and Tandon presented in \cite{Chang2018OnTC}. In summary, our work shows that appropriate matrix partition is of importance in enabling rate-efficient, straggler-robust and secure two-sided distributed matrix computation. 

\appendices 
\section{Close-to-optimal Solution of Optimization Problem \eqref{opt}}
\label{sec:close_opt_sol}

\textcolor{black}{Next, we propose a close-to-optimal solution to} the optimization problem \eqref{opt}. \textcolor{black}{To establish this solution, we need the following lemmas.}

\begin{lemma}
	\label{claim1}
	For every optimal solution of the optimization problem \eqref{opt}, there is at least one maximizing pair denoted by $(r^{\star}_{\bm A},r^{\star}_{\bm B})$ which satisfies $r^{\star}_{\bm A}\:\textcolor{black}{\geq}\:r^{\star}_{\bm B}$.
\end{lemma}

\begin{proof}
	\alert{Proof by contradiction. Suppose that the maximizing pair $(r^{\star}_{\bm A},r^{\star}_{\bm B})$ satisfies $r^{\star}_{\bm A}<r^{\star}_{\bm B}$}. The associated number of exploited servers is then given by
	\begin{align*}
	Q^{\star}(r^{\star}_{\bm A},r^{\star}_{\bm B})=\big(r^{\star}_{\bm A}+\ell\big)\big(r^{\star}_{\bm B}+1\big)-1.
	\end{align*}
	\alert{On the other hand, the associated number of exploited servers for the \emph{inverted} pair $(r^{\star}_{\bm B},r^{\star}_{\bm A})$ corresponds to} 
	\begin{align*}
	Q'(r^{\star}_{\bm B},r^{\star}_{\bm A})=\big(r^{\star}_{\bm B}+\ell\big)\big(r^{\star}_{\bm A}+1\big)-1.
	\end{align*}
	Subtracting $Q'$ from $Q^{\star}$ gives
	\begin{align*}
	Q'-Q^{\star}=(r^{\star}_{\bm A}-r^{\star}_{\bm B})(\ell-1)\leq 0.
	\end{align*}
	\alert{When $Q'-Q^{\star}\leq 0$, we have 
		\begin{align}
		\notag
		Q'-Q^{\star}\leq 0\Leftrightarrow &\: Q'\leq Q^{\star}\\
		\label{ineq1}\Leftrightarrow & R(r^{\star}_{\bm B},r^{\star}_{\bm A})=\frac{r^{\star}_{\bm A}r^{\star}_{\bm B}}{Q'}\geq \frac{r^{\star}_{\bm A}r^{\star}_{\bm B}}{Q^{\star}}=R^{\star}(r^{\star}_{\bm A},r^{\star}_{\bm B}). 
		\end{align}\noindent We infer from inequality \eqref{ineq1} that the \emph{inverted} pair $(r^{\star}_{\bm B},r^{\star}_{\bm A})$ attains a higher rate than $(r^{\star}_{\bm A},r^{\star}_{\bm B})$. This is in contradiction with the assumption that $(r^{\star}_{\bm A},r^{\star}_{\bm B})$ is a maximizing pair.} \\
\end{proof}

\begin{lemma}
	\label{claim2}
	\alert{When $\ell_{\max}=\lfloor \frac{N-1}{2}\rfloor$ and $N\geq 3$, $(r^{\star}_{\bm A},r^{\star}_{\bm B})=(1,1)$ is an \emph{unique} maximizing pair of the optimization problem \eqref{opt}.}
\end{lemma}
\begin{proof} Define
	\begin{align*}
	\begin{split}
	\ell_{\max}=\biggl\lfloor \frac{N-1}{2}\biggr\rfloor=
	\begin{cases}
	\frac{N}{2}-1 \qquad & N\text{ even} \\
	\frac{N-1}{2} \qquad & N\text{ odd}
	\end{cases}.
	\end{split}
	\end{align*}
	\alert{We set $(r_{\bm A},r_{\bm B})=(1+a,1+b),\text{ where } a,b\in \mathbb{Z}^{+}\cup\{0\}$, such that the number of exploited servers equals} 
	\begin{align*}
	\begin{split}
	Q_{N,\ell_{\max}}=&\big(r_{\bm A}+\ell_{\max}\big)\left(r_{\bm B}+1\right)-1\\=&\big(a+\ell_{\max}+1\big)\left(b+2\right)-1\\
	=&\begin{cases}
	N-1+2a+b+ab+\ell_{\max}b \qquad & N\text{ even} \\
	N+2a+b+ab+\ell_{\max}b \qquad & N\text{ odd}
	\end{cases}.
	\end{split}
	\end{align*}
	\alert{The only feasible pair $(a,b)$ satisfying the inequality constraint of the optimization problem \eqref{opt} is $(a,b)=(0,0)$. Therefore the only maximizing pair of the optimization problem \eqref{opt} is $(r^{\star}_{\bm A},r^{\star}_{\bm B})=(1,1)$.}
\end{proof}
\begin{definition}
	\label{feasible}
	\alert{For a given $\ell$ and $N$, $(r_{\bm A}, r_{\bm B})$ is a \emph{strongly feasible} pair of the optimization problem \eqref{opt} if and only if 
		\begin{enumerate}[label=($\roman*$)]
			\item it satisfies the inequality constraint $Q_{N,\ell}(r_{\bm A}, r_{\bm B})\leq N$,
			\item and there exists no \emph{feasible} pair $(r'_{\bm A}, r_{\bm B})$ or $(r_{\bm A}, r'_{\bm B})$ with $r'_{\bm A}\geq r_{\bm A}$ \textcolor{black}{or} $r'_{\bm B}\geq r_{\bm B}$.
		\end{enumerate}} 
\end{definition}
	\begin{lemma}
		\label{feasibility}
		\alert{Every maximizing pair $(r^{\star}_{\bm A}, r^{\star}_{\bm B})$ of the optimization problem \ref{opt} satisfies the strong feasibility condition.} 
	\end{lemma}
	\begin{proof}
		The rate for any pair $(r_{\bm A}, r_{\bm B})\in\mathbb{Z}_{2}^{+}$ is given by
		\begin{align}
		R(r_{\bm A}, r_{\bm B})=\frac{r_{\bm A}r_{\bm B}}{r_{\bm A}r_{\bm B}+\ell r_{\bm B}+r_{\bm A}+\ell-1}= \frac{1}{1+\frac{\ell}{r_{\bm A}}+\frac{1}{r_{\bm B}}+\frac{\ell-1}{r_{\bm A}r_{\bm B}}}.
		\label{rateEq}
		\end{align}
		\alert{Suppose by contradiction that the \textcolor{black}{maximizing} pair $(r^{\star}_{\bm A}, r^{\star}_{\bm B})$ is \emph{not} strongly feasible, i.e., it does not satisfy condition $(ii)$ of Definition \ref{feasibility}. Thus, $r_{\bm A}$ or $r_{\bm B}$ can be increased to values above $r^{\star}_{\bm A}$ or $r^{\star}_{\bm B}$ without violating the inequality constraint of the optimization problem. 
			An increase of $r_{\bm A}$ or $r_{\bm B}$ leads to an increase in the rate (cf. Eq. \ref{rateEq}). This contradicts that $(r^{\star}_{\bm A}, r^{\star}_{\bm B})$ is a maximizing pair of the optimization problem \ref{opt}.}
	\end{proof}
	
	\begin{lemma}
		\label{claim3}
		\alert{Let $(r_{\bm A}, r_{\bm B})$ be a strongly feasible pair. When $\ell$ decreases by one $(\ell\leftarrow\ell-1)$ and we simultaneously increase $r_{\bm A}$ by one $(r_{\bm A}\leftarrow r_{\bm A}+1)$ while keeping $r_{\bm B}$ constant $(i)$ has no effect on the number of exploited servers and keeps it at $\tilde{Q}\triangleq Q_{\ell}(r_{\bm A}, r_{\bm B})$, $(ii)$ generates a new strongly feasible pair $(r_{\bm A}+1, r_{\bm B})$ at $\ell-1$ and $(iii)$ increases the rate additively by $\frac{r_{\bm B}}{\tilde{Q}}$.}
	\end{lemma}
	\begin{proof}
		Consider the pairs $\left(r_{1,\bm A},r_{1, \bm B}\right)$ and $\left(r_{2,\bm A},r_{2, \bm B}\right)$, where $r_{2,\bm A}=r_{1,\bm A}+1$ and \textcolor{black}{$r_{2,\bm B}=r_{1, \bm B}$}. The number of exploited servers for the pair $\left(r_{2,\bm A},r_{2, \bm B}\right)$ with \alert{$\ell_{2}=\ell_{1}-1$} colluding servers is given by
		\begin{align}
		\label{eq20}
		\tilde{Q}\triangleq Q_{\ell_2}(r_{2,\bm A},r_{2, \bm B})=\left(r_{1, \bm A}+\ell_1+1-1\right)\left(r_{1,\bm B}+1\right)-1=Q_{\ell_1}(r_{1, \bm A},r_{1, \bm B}).
		\end{align}
		\alert{From $Q_{\ell_2}=Q_{\ell_1}$, $(i)$ and $(ii)$ of Lemma \ref{claim3} readily follow. The rate associated with the pair $\left(r_{2,\bm A},r_{2, \bm B}\right)$ then becomes 
			\begin{align*}
			R_{\ell_2}(r_{2,\bm A},r_{2, \bm B}) =&
			\frac{r_{2, \bm A}r_{2,\bm B}}{Q_{\ell_2}(r_{2,\bm A},r_{2, \bm B})}=
			\frac{(r_{1, \bm A}+1)r_{1,\bm B}}{Q_{\ell_1}(r_{1, \bm A},r_{1, \bm B})}\\
			=&\frac{r_{1, \bm A}r_{1,\bm B}}{\tilde{Q}}+\frac{r_{1,\bm B}}{\tilde{Q}}.
			\end{align*}} 
	\end{proof}
	\begin{lemma}
		\label{claim4}
		\alert{Suppose that $(r_{1,\bm A}, r_{1,\bm B})$ and $(r_{2,\bm A}, r_{2,\bm B})$ are two strongly feasible pairs, where $r_{2,\bm B}\geq r_{1,\bm B}$ \textcolor{black}{and $r_{1,\bm A}+1\geq r_{1,\bm B}$}. Decreasing $\ell$ by one $(\ell\leftarrow\ell-1)$ and simultaneously increasing $r_{1,\bm A}$ $(r_{1,\bm A}\leftarrow r_{1,\bm A}+1)$ and $r_{2,\bm A}$ $(r_{2,\bm A}\leftarrow r_{2,\bm A}+1)$ by one while not changing $r_{1,\bm B}$ and $r_{2,\bm B}$ results in an increase of the rate for both of the pairs $(r_{1,\bm A}, r_{1,\bm B})$ and $(r_{2,\bm A}, r_{2,\bm B})$. The additive increase in the rate for the pair $(r_{2,\bm A}, r_{2,\bm B})$ is larger than for the pair $(r_{1,\bm A}, r_{1,\bm B})$.} 
	\end{lemma}
	\begin{proof}
		\textcolor{black}{We have $r_{1,\bm A}+1\geq r_{1,\bm B}$ so that}
		\begin{align}
		\notag
		&\textcolor{black}{r_{1,\bm B}^2\leq (r_{1,\bm A}+1) r_{1,\bm B} \leq (r_{1,\bm A}+\ell)r_{1,\bm B} \leq (r_{1,\bm A}+\ell) (r_{1,\bm B}+1)-1 = Q_{\ell}(r_{1,\bm A},r_{1,\bm B})} \\
		\label{ineq2}
		\Leftrightarrow \qquad & r_{1,\bm B}^{2}+r_{1,\bm B} Q_{\ell}(r_{1,\bm A},r_{1,\bm B}) \leq r_{1,\bm B} Q_{\ell}(r_{1,\bm A},r_{1,\bm B})+ Q_{\ell}(r_{1,\bm A},r_{1,\bm B})\nonumber\\\Leftrightarrow\qquad &\frac{r_{1,\bm B}}{Q_{\ell}(r_{1,\bm A},r_{1,\bm B})}\leq \frac{r_{1,\bm B}+1}{r_{1,\bm B}+Q_{\ell}(r_{1,\bm A},r_{1,\bm B})}
		\end{align} follows. 
		Since $(r_{1,\bm A}, r_{1,\bm B})$ is a strongly feasible pair (cf. Definition \ref{feasible}), neither $r_{1,\bm A}$ nor $r_{1,\bm B}$ can increase while the other element of the pair $(r_{1,\bm A}, r_{1,\bm B})$ remains constant. Recall that
		\begin{align}
		\label{eq22}
		Q_{\ell}(r_{1,\bm A},r_{1,\bm B})=\left(r_{1, \bm A}+\ell\right)\left(r_{1,\bm B}+1\right)\alert{-1}\leq N.
		\end{align}
		\alert{Incrementing $r_{1, \bm A}$ by one increases $Q_{\ell}$ by $r_{1,\bm B}+1$.} Similarly, increasing \alert{$r_{1, \bm B}$} by one  enlarges $Q_{\ell}$ by $r_{1,\bm A}+\ell$. Moreover, $r_{1, \bm A}\geq r_{1, \bm B}$ and $\ell\geq 1$ implies  
		\begin{align*}
		r_{1, \bm A}+\ell\geq r_{1, \bm B}+1.
		\end{align*}
		\alert{Therefore, strong feasibility along with above observation suggests that $Q_{\ell}$ is at least $N-r_{1, \bm B}$.} We remind the reader that as long as $Q_{\ell}< N-r_{1, \bm B}$, the strong feasibility assumption is violated.  
		As a result, the number of exploited servers for the strongly feasible pair $(r_{1,\bm A}, r_{1,\bm B})$ is lower bounded according to 
		\begin{align*}
		N-r_{1,\bm B}\leq Q_{\ell}(r_{1,\bm A},r_{1,\bm B}).
		\end{align*}
		On the other hand, the number of exploited servers for the strongly feasible pair \alert{$(r_{2,\bm A}, r_{2,\bm B})$} is bounded from above by 
		\begin{align*}
		Q_{\ell}(r_{2,\bm A},r_{2,\bm B}) \leq N.
		\end{align*}
		Therefore 
		\begin{align}
		&Q_{\ell}(r_{2,\bm A},r_{2,\bm B})-Q_{\ell}(r_{1,\bm A},r_{1,\bm B})\leq r_{1,\bm B}\nonumber\\\label{ineq3}\Leftrightarrow\qquad & Q_{\ell}(r_{2,\bm A},r_{2,\bm B}) \leq Q_{\ell}(r_{1,\bm A},r_{1,\bm B}) + r_{1,\bm B}. 
		\end{align}
		\alert{Combining inequalities \eqref{ineq2}, \eqref{ineq3} and $r_{1,\bm B}\leq r_{2,\bm B}$, we get  
			\begin{align*}
			\frac{r_{1,\bm B}}{Q_{\ell}(r_{1,\bm A},r_{1,\bm B})}\stackrel{\eqref{ineq2}}{\leq}\frac{r_{1,\bm B}+1}{r_{1,\bm B}+Q_{\ell}(r_{1,\bm A},r_{1,\bm B})}\stackrel{\eqref{ineq3}}{\leq} \frac{r_{1,\bm B}+1}{Q_{\ell}(r_{2,\bm A},r_{2,\bm B})} \leq \frac{r_{2,\bm B}}{Q_{\ell}(r_{2,\bm A},r_{2,\bm B})}.
			\end{align*}} \alert{Hereby, $\frac{r_{1,\bm B}}{Q_{\ell}(r_{1,\bm A},r_{1,\bm B})}$ and $\frac{r_{2,\bm B}}{Q_{\ell}(r_{2,\bm A},r_{2,\bm B})}$ are the corresponding terms by which the rate increases after subjecting the pairs $(r_{1,\bm A}, r_{1,\bm B})$ and $(r_{2,\bm A}, r_{2,\bm B})$ to the mapping of Lemma \ref{claim3}.}
	\end{proof}
	
	\begin{lemma}
		\label{claim5}
		\alert{Consider the optimization problem \eqref{opt} for a constant $N$ and two consecutive values of $\ell$ denoted by $\ell_{1}$ and $\ell_{2}=\ell_{1}-1$, respectively. The optimal variables for these two problems, represented by $\bm{r}_{\ell_1}^{\star}\triangleq(r^{\star}_{1,\bm A},r^{\star}_{1, \bm B})$ if $\ell=\ell_1$ and $\bm{r}_{\ell_2}^{\star}\triangleq(r^{\star}_{2,\bm A},r^{\star}_{2, \bm B})$ if $\ell=\ell_2$, have a specific relation. That is, if $\bm{r}_{\ell_1}^{\star}$ is known, there are just two possibilities for  $\bm{r}_{\ell_2}^{\star}$:
			\begin{enumerate}[label=($\roman*$)]
				\item $\bm{r}_{\ell_2}^{\star}=\big(r^{\star}_{1,\bm A}+1,r^{\star}_{1, \bm B}\big)$,
				\item $\bm{r}_{\ell_2}^{\star}$
				satisfies $r^{\star}_{2,\bm A}\leq r^{\star}_{1,\bm A}$ and $r^{\star}_{2,\bm B}>r^{\star}_{1,\bm B}$.
			\end{enumerate}}
		\end{lemma}
		\begin{proof} In the following, we go through three possibilities in the choice of $r^{\star}_{2, \bm B}$ in comparison to $r^{\star}_{1, \bm B}$: $(i) \: r^{\star}_{2, \bm B}=r^{\star}_{1, \bm B}, (ii)\: r^{\star}_{2, \bm B}<r^{\star}_{1, \bm B}$ and $(iii)\: r^{\star}_{2, \bm B}>r^{\star}_{1, \bm B}$.
			\begin{enumerate}[label=($\roman*$)]
				\item \alert{If $r^{\star}_{2, \bm B}=r^{\star}_{1, \bm B}$, choose $r^{\star}_{2, \bm A}=r^{\star}_{1, \bm A}+a,a\in \mathbb{Z}$. Then, the number of exploited servers becomes:
					\begin{align}
					Q_{\ell_2}^{\star}(\bm{r}_{\ell_2}^{\star})&=\big(r^{\star}_{2, \bm A}+\ell_2 \big)\big(r^{\star}_{2, \bm B}+1\big)-1\notag \\\notag
					&=\big(r^{\star}_{1, \bm A}+a+\ell_1 -1\big)\big(r^{\star}_{1, \bm B}+1\big)-1\\ 
					\label{eqNew}
					& =\big(r^{\star}_{1, \bm A}+a'+\ell_1\big)\big(r^{\star}_{1, \bm B}+1\big)-1\leq N,
					\end{align} where $a'=a-1,a'\in\mathbb{Z}$.   
					Since the pair $\bm{r}_{\ell_1}^{\star}$ is optimal, according to Lemma \ref{feasibility}, it must be strongly feasible. Due to the strong feasibility of the pairs $\bm{r}_{\ell_1}^{\star}$ and $\bm{r}_{\ell_2}^{\star}$, we have $a'=0$.} \alert{The other possibilities $a'\in \mathbb{Z}^{+}$ or $a'\in \mathbb{Z}^{-}$ are sub-optimal. First, if \eqref{eqNew} is satisfied for $a'\in \mathbb{Z}^{+}$ contradicts with the strong feasibility assumption of the pair $\bm{r}_{\ell_1}^{\star}$.
					Second, if $a'\in \mathbb{Z^{-}}$, there exists another strongly feasible pair $(r^{\star}_{2,\bm A}-a',r^{\star}_{2, \bm B})$ for $\ell=\ell_1$. Then, this strongly feasible pair fulfills the inequality 
					\begin{align*}
					Q_{\ell_1}^{\star}(\bm{r}_{\ell_1}^{\star})&=\big(r^{\star}_{1, \bm A}+\ell_1 \big)\big(r^{\star}_{1, \bm B}+1\big)-1\\
					&=\big(r^{\star}_{2, \bm A}-a+\ell_2 +1 \big)\big(r^{\star}_{2, \bm B}+1\big)-1\\
					&=\big(r^{\star}_{2, \bm A}-a'+\ell_2\big)\big(r^{\star}_{2, \bm B}+1\big)-1\leq N,
					\end{align*} where $a'=a-1$ and $a'\in \mathbb{Z^{-}}$.
					However, this is in conflict with the pair $\bm{r}_{\ell_2}^{\star}$ being strongly feasible. In summary, this establishes possibility $(i)$ of Lemma \ref{claim5}.}
				
				\item \alert{If $r^{\star}_{2, \bm B}<r^{\star}_{1, \bm B}$, we choose $r^{\star}_{2,\bm A}$ so that the pair $\tilde{\bm{r}}_{\ell_1}\triangleq(r^{\star}_{2,\bm A}-1,r^{\star}_{2, \bm B})$ is strongly feasible for $\ell_1$. Recall from Lemma \ref{claim3} that the pair $\bm{r}_{\ell_2}^{\star}$ is also strongly feasible if $\ell_2=\ell_1-1$ and $N$ being constant. Simultaneously, the pair $\bm{r}^{\star}_{\ell_1}$ maximizes the rate $R_{N,\ell_1}$. Thus, we have
					\begin{align}
					\label{ineqRate}
					R_{N,\ell_1}(\tilde{\bm{r}}^{\star}_{\ell_1}) \leq R_{N,\ell_1}(\bm{r}^{\star}_{\ell_1}).
					\end{align} We denote the rate increase from the rate pair $\tilde{\bm{r}}_{\ell_1}$ to $\bm{r}^{\star}_{\ell_2}$ when $\ell$ decreases by one $(\ell_1\leftarrow\ell_2)$ by $\Delta_{N,\ell_1\leftarrow\ell_2}(\tilde{\bm{r}}_{\ell_1}\leftarrow\bm{r}^{\star}_{\ell_2})$. Similarly, $\Delta_{N,\ell_1\leftarrow\ell_2}(\bm{r}^{\star}_{\ell_1}\leftarrow\tilde{\bm{r}}_{\ell_2})$ refers to the rate increase from $\bm{r}^{\star}_{\ell_1}$ to $\tilde{\bm{r}}_{\ell_2}\triangleq(r^{\star}_{1,\bm A}+1,r^{\star}_{1, \bm B})$. Thus, the overall achievable rates at pairs $\bm{r}^{\star}_{\ell_2}$ and $\tilde{\bm{r}}_{\ell_2}$ correspond to
					\begin{align}
					\label{eqRate}
					\begin{split}
					R_{N,\ell_2}(\bm{r}^{\star}_{\ell_2})=R_{N,\ell_1}(\tilde{\bm{r}}_{\ell_1})+\Delta_{N,\ell_1\leftarrow\ell_2}(\tilde{\bm{r}}_{\ell_1}\leftarrow\bm{r}^{\star}_{\ell_2})\\
					R_{N,\ell_2}(\tilde{\bm{r}}_{\ell_2})=R_{N,\ell_1}(\bm{r}^{\star}_{\ell_1})+\Delta_{N,\ell_1\leftarrow\ell_2}(\bm{r}^{\star}_{\ell_1}\leftarrow\tilde{\bm{r}}_{\ell_2})
					\end{split}.
					\end{align} Due to Lemma \ref{claim4}, we have $\Delta_{N,\ell_1\leftarrow\ell_2}(\bm{r}^{\star}_{\ell_1}\leftarrow\tilde{\bm{r}}_{\ell_2})\geq\Delta_{N,\ell_1\leftarrow\ell_2}(\tilde{\bm{r}}_{\ell_1}\leftarrow\bm{r}^{\star}_{\ell_2})$. Consequently, we infer from \eqref{ineqRate} and \eqref{eqRate} that $R_{N,\ell_2}(\tilde{\bm{r}}_{\ell_2})\geq R_{N,\ell_2}(\bm{r}^{\star}_{\ell_2})$. However, this violates the assumption of optimality at $\bm{r}^{\star}_{\ell_2}$. Thus, $\bm{r}^{\star}_{\ell_2}$ cannot be a maximizing pair if $r^{\star}_{2, \bm B}< r^{\star}_{1, \bm B}$.} 
				\item \alert{If $r^{\star}_{2, \bm B}>r^{\star}_{1, \bm B}$, the number  of exploited servers $Q_{\ell_2}^{\star}$ for the pair $\bm{r}^{\star}_{\ell_2}$ can be lower bounded according to
					\begin{align}
					\begin{split}
					\label{q2}
					Q^{\star}_{\ell_2}(\bm{r}^{\star}_{\ell_2})&=\big(r^{\star}_{2, \bm A}+\ell_2\big)\big(r^{\star}_{2, \bm B}+1\big)-1\\
					&=\big(r^{\star}_{2, \bm A}-1+\ell_{1}\big)\big(r^{\star}_{2, \bm B}+1\big)-1\\
					&\geq \big(r^{\star}_{2, \bm A}-1+\ell_1\big)\big(r^{\star}_{1, \bm B}+2\big)-1.
					\end{split}
					\end{align} Since $Q^{\star}_{\ell_2}(\bm{r}^{\star}_{\ell_2})\leq N$, we infer that
					\begin{align}
					\label{eq25}
					\big(r^{\star}_{2, \bm A}-1+\ell_1\big)\big(r^{\star}_{1, \bm B}+2\big)-1\leq N. 
					\end{align} Since the pair $\bm{r}^{\star}_{\ell_1}$ is strongly feasible, we deduce from \eqref{eq25} that 
					\begin{align*}
					r^{\star}_{2,\bm A}-1\leq r^{\star}_{1,\bm A}-1 \Leftrightarrow r^{\star}_{2,\bm A}\leq r^{\star}_{1,\bm A}.
					\end{align*} Thus, possibility \textcolor{black}{$(iii)$} of Lemma \ref{claim5} is shown.}  
			\end{enumerate}

		\end{proof}
		
		\alert{Beginning from $\ell_{\max}=\lfloor \frac{N-1}{2}\rfloor$, where $\bm{r}_{\ell_{\max}}^{\star}=(1,1)$ (cf. Lemma \ref{claim2}), we seek to determine the optimal $\bm{r}^{\star}_{\ell-1}\triangleq(r^{\star}_{{\ell-1},\bm A},r^{\star}_{{\ell-1},\bm B})$ from $\bm{r}^{\star}_{\ell}\triangleq(r^{\star}_{{\ell},\bm A},r^{\star}_{{\ell},\bm B})$. To this end, we exploit Lemma \ref{claim5}, which states that if $\ell$ decreases by $1$, either $(i)$ $r^{\star}_{{\ell},\bm B}=r^{\star}_{{\ell-1},\bm B}$ or $(ii)$ $r^{\star}_{{\ell},\bm B}>r^{\star}_{{\ell-1},\bm B}$. Specifically, beginning from $\ell_{\max}$, we iteratively move backwards towards $\ell_{\min}=1$ and \emph{estimate} the values of $\ell$ at which $r^{\star}_{\ell,\bm B}$ increases in comparison to previous iterates. This helps us to determine close-to-optimal estimates $\hat{r}_{\ell,\bm B}$ for a given $\ell$ at constant $N$. Using these estimates $\hat{r}_{\ell,\bm B}$, one can solve for $\hat{r}_{\ell,\bm A}$ using the inequality of the optimization problem \eqref{opt}.} 
		
		Let us start with the process of finding $\ell$, where $r^{\star}_{\ell,\bm B}$ must increase compared to $r^{\star}_{\ell-1,\bm B}$. In other words, we restrict ourself to track at which values of $\ell\in[\ell_{\min}:\ell_{\max}]$, $r^{\star}_{\ell,\bm B}$ changes. For notational simplicity, we rewrite $r^{\star}_{\ell_{m},\bm B}$ by $m$ and denote its corresponding pair by $r^{\star}_{\ell_{m},\bm A}$. Here, $\ell_{m}$ denotes the first iterate\footnote{Recall that we move backwards from $\ell_{\max}$ to $\ell_{\min}$.}, or largest $\ell$, for which the optimal $r^{\star}_{\bm B}$ changes from $m'$ to $m$, or mathematically, 
		\begin{align*}
		\ell_{m}=\max\{\ell\in[\ell_{\min}:\ell_{\max}]\:|\:r^{\star}_{\ell,\bm B}=m\}.
		\end{align*} According to $(ii)$ of Lemma \ref{claim5}, the difference of $r_{\bm B}$-values at neighboring $\ell$-values -- $\ell_{m}$ and $\ell_{m}-1$ -- i.e.,  
		\begin{align*}
		r^{\star}_{\ell_{m},\bm B}-r^{\star}_{\ell_{m}-1,\bm B}=m-m'
		\end{align*} is lower-bounded by $1$. Recall that this difference does not have to be necessarily one. However, if we \emph{relax the integer assumption} on $\ell_{m}$ and $\ell_{m-1}$, we can find an $\ell_{m}$ such that $m=m'+1$. This is shown in Fig. \ref{fig:non_integer} (by the step functions in the interval $\ell_m\leq\ell\leq\ell_{m'}$). \alertv{We will see at the end of this appendix that the values of $m$ which are not associated to (positive) integer-valued $\ell_m$ are excluded from the results by using the integer assumption. Furthermore, the non integer values of $\hat{r}_{\ell, \bm B}$ are avoided by applying the ceiling function $\lceil\cdot\rceil$.} We define
		\begin{align*}
		d_{m}=|\ell_{m'}-\ell_{m}|
		\end{align*} as the required number of steps in $\ell$ needed to change $r^{\star}_{\bm B}$ from $m'$ to $m$. In the sequel, we neglect that $\ell$, $r_{\bm A}$ and $d_{m}$ are integer numbers. In the notation this is accounted by using $\hat{\ell}$, $\hat{r}_{\bm A}$ and $\hat{d}_{m}$, respectively. 
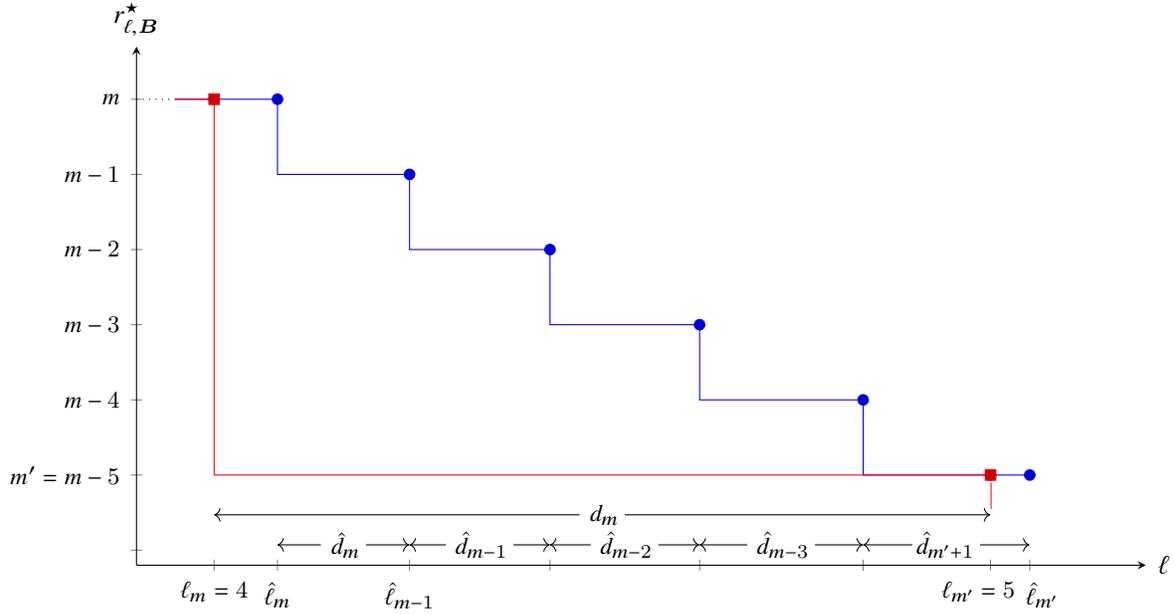
\begin{figure}[t]
	\centering
	\begin{tikzpicture}
	\begin{axis}
	[
	,width=15cm
	,height=10cm
	,xlabel=\small $\ell$
	,ylabel=\small $r^{\star}_{\ell,\bm B}$
	,xmin=3.90, xmax=5.2, ymin=42.8, ymax=49.7,
	,axis x line=middle
	,axis y line=middle
	,enlargelimits=false
	,y=1cm
	,xlabel style={
		anchor=west,
		at={(ticklabel* cs:0.975)},
		xshift=10pt
	}
	,       ylabel style={
		anchor=south,
		at={(ticklabel* cs:0.95)},
		yshift=10pt
	}
	,xtick=data,
	,xticklabels={\footnotesize $\ell_{m}=4$,\footnotesize $\hat{\ell}_{m}$,\footnotesize $\hat{\ell}_{m-1}$,,,,\footnotesize $\ell_{m'}=5\quad$,\footnotesize $\quad\hat{\ell}_{m'}$},
	,yticklabels={,,,\footnotesize $m'=m-5$
		,\footnotesize $m-4$,\footnotesize $m-3$,\footnotesize $m-2$,\footnotesize $m-1$,\footnotesize $m$}
	]
	\addplot+[blue][only marks] coordinates
	{(4,49) (4.0816,49) (4.2517,48) (4.4326,47) (4.6253,46) (4.8359,45) (5,44) (5.0505,44)};
	\addplot+[red][only marks] coordinates
	{(4,49) (5,44)};
	\addplot+[blue][mark=none] coordinates
	{(3.95,49) (4.0816,49) (4.0816,48) (4.2517,48) (4.2517,47) (4.4326,47) (4.4326,46) (4.6253,46)(4.6253,45) (4.8359,45) (4.8359,44) (5.0505,44)
	};
	\addplot+[red][mark=none] coordinates
	{(3.95,49) (4,49) (4,44) (5,44)
	};
	\end{axis}
	
	\draw [<-,yshift=3.7ex] (11.35,0) -- (6.49,0) node[left=-0.4ex] {{\footnotesize $d_{m}$}};
	\draw [->,yshift=3.7ex] (5.89,0) -- (1.03,0);

	\draw [<-,yshift=1.5ex] (3.6275,0) -- (3.0498,0) node[left=-0.4ex] {{\footnotesize $\hat{d}_{m}$}};
	\draw [->,yshift=1.5ex] (2.4498,0) -- (1.8721,0);
	\draw [<-,yshift=1.5ex] (5.4944,0) -- (5.0110,0) node[left=-0.4ex] {{\footnotesize $\hat{d}_{m-1}$}};
	\draw [->,yshift=1.5ex] (4.1109,0) -- (3.6275,0);
	
	\draw [<-,yshift=1.5ex] (7.4831,0) -- (6.9387,0) node[left=-0.4ex] {{\footnotesize $\hat{d}_{m-2}$}};
	\draw [->,yshift=1.5ex] (6.0387,0) -- (5.4944,0);
	
	\draw [<-,yshift=1.5ex] (9.6564,0) -- (9.0197,0) node[left=-0.4ex] {{\footnotesize $\hat{d}_{m-3}$}};
	\draw [->,yshift=1.5ex] (8.1197,0) -- (7.4831,0);
	
	\draw [<-,yshift=1.5ex] (11.8712,0) -- (11.2138,0) node[left=-0.4ex] {{\footnotesize $\hat{d}_{m'+1}$}};
	\draw [->,yshift=1.5ex] (10.3138,0) -- (9.6564,0);
	\draw [dotted] (0,6.2) -- (0.55,6.2);
	
	\draw [red] (11.36,1.1) -- (11.36,0.75);
	
	\end{tikzpicture}
	\caption{\small Illustrative behavior of $r_{\ell,\bm B}$ with respect to $\hat{\ell}_{m}$ for $N=10000$, $\ell_{m}=4$, $m=r^{\star}_{\ell_{m}, \bm B}=49$ and $m'=r^{\star}_{\ell_{m'}, \bm B}=44$. In our approximation we allow for a non-integer relaxation in $\ell_m$ which we denote by $\hat{\ell}_m$. This allows us to associate values $\hat{\ell_j}$ to $r_{\ell,\bm B}\in[m':m]$. According to the figure, $\hat{\ell}_{m}$ is given by $\hat{\ell}_{m}=\hat{\ell}_{m'}-\sum\limits_{i=m'+1}^{m}\hat{d}_i$ or $\hat{\ell}_{m}=\ell_{\max}-\sum\limits_{i=2}^{m}\hat{d}_i$.} 
	\label{fig:non_integer}
\end{figure}

		\textcolor{black}{Due to Lemma \ref{claim3}, the rate of the optimal pair $( r^{\star}_{\ell_{m},\bm A},m)$ is larger than the rate of the sub-optimal pair $(r^{\star}_{\ell_{m}+1,\bm A}+1,m')$.} 
		This translates to the inequality
		\begin{align}
		\label{ineq4}
		\textcolor{black}{\frac{m'\big(r^{\star}_{\ell_{m}+1,\bm A}+1\big)}{Q^{\star}_{\ell_{m}+1}} \leq \frac{mr^{\star}_{\ell_{m},\bm A}}{Q^{\star}_{\ell_{m}}}}.
		\end{align} 
		Applying $(i)$ of Lemma \ref{claim5} on the pairs  $(r^{\star}_{\ell_{m'},\bm A},m')$ and $(r^{\star}_{\ell_{m}+1,\bm A},m')$ leads to 
		\begin{align*}
		r^{\star}_{\ell_{m}+1,\bm A}=r^{\star}_{\ell_{m'},\bm A}+d_{m}-1.
		\end{align*}
		We can now reformulate \eqref{ineq4} as
		\begin{align}
		\label{ineq5}
		\textcolor{black}{\frac{m'(r^{\star}_{\ell_{m'},\bm A}+d_{m})}{Q^{\star}_{\ell_{m}+1}} \leq \frac{mr^{\star}_{\ell_{m},\bm A}}{Q^{\star}_{\ell_{m}}}}.
		\end{align} \textcolor{black}{Assuming that the number of exploited servers is almost constant, i.e.,  $Q^{\star}_{\ell_{m}+1}\approx Q^{\star}_{\ell_{m}}$ allows us to ignore the denominators of the fractions in \eqref{ineq5}}. Further, we apply the \emph{integer relaxation} such that $m'=m-1$. These approximations allow us to transform \eqref{ineq5} to an equality given by\footnote{Due to these approximations, we replace the $"\star"$ superscript with $"\:\:\widehat{}\:\:"$.} 
		\begin{align}
		\notag &\left(m-1\right)(\hat{r}_{\ell_{m-1},\bm A}+\hat{d}_{m})= m\hat{r}_{\ell_{m},\bm A}\\
		\label{eq5}\Leftrightarrow \quad& \hat{d}_{m}=\frac{m}{m-1}\hat{r}_{\ell_{m},\bm A}-\hat{r}_{\ell_{m-1},\bm A},
		\end{align} where $m\in\mathbb{Z}^{+}\setminus\{1\}$. Next, we calculate $\hat{Q}_{\ell_{m}}$ for the pair $(\hat{r}_{\ell_{m},\bm A},m)$ as follows
		\begin{align}
		\notag \hat{Q}_{\ell_{m}}&=(\hat{r}_{\ell_{m},\bm A}+\ell_{m})\left(m+1\right)-1\\
		\notag &=\Big(\hat{r}_{\ell_{m},\bm A}+\ell_{\max}-\sum_{i=2}^{m} \hat{d}_{i}\Big)\left(m+1\right)-1\\
		\notag &=\hat{r}_{\ell_{m},\bm A}\left(m+1\right)+\ell_{\max}\left(m+1\right)-\bigg(\sum_{i=2}^{m} \hat{d}_{i}\bigg)\left(m+1\right)-1\\
		\notag &\stackrel{(a)}{\approx}\hat{r}_{\ell_{m},\bm A}\left(m+1\right)+\frac{N}{2}\left(m-1+2\right)-\left(\sum_{i=2}^{m} \hat{d}_{i}\right)\left(m+1\right)-1\\
		\notag &=N+\hat{r}_{\ell_{m},\bm A}\left(m+1\right)+\frac{N}{2}\left(m-1\right)-\left(\sum_{i=2}^{m} \hat{d}_{i}\right)\left(m+1\right)-1\\&\stackrel{(b)}{\textcolor{violet}{=}}N-1\nonumber,
		\end{align} where $(a)$ and $(b)$ are due to the approximations $\ell_{\max}\approx\frac{N}{2}$ and  $\hat{Q}_{\ell_{m}}\approx N-1$, respectively. It is easy to conclude from above \textcolor{black}{equality} that 
		\begin{align}
		\label{eq1} \hat{r}_{\ell_{m},\bm A}\left(m+1\right)+\frac{N}{2}\left(m-1\right)-\left(\sum_{i=2}^{m} \hat{d}_{i}\right)\left(m+1\right)=0.
		\end{align}
		From Eq. \eqref{eq1}, we get
		\begin{align}
		\label{eq2}\sum_{i=2}^{m} \hat{d}_{i}=\frac{(m-1)}{(m+1)}\frac{N}{2}+\hat{r}_{\ell_{m},\bm A}.
		\end{align} Similarly, we can approximate $\hat{Q}_{\ell_{m-1}}$ by the same approach such that
		\begin{align}
		\label{eq3}\sum_{i=2}^{m-1} \hat{d}_{i}=\frac{(m-2)}{m}\frac{N}{2}+\hat{r}_{\ell_{m-1},\bm A}.
		\end{align}
		Subtracting \eqref{eq3} from \eqref{eq2} gives
		\begin{align}
		\label{eq4}
		\begin{split}
		\hat{d}_{m}=&\left(\frac{m-1}{m+1}-\frac{m-2}{m}\right)\frac{N}{2}+\hat{r}_{\ell_{m},\bm A}-\hat{r}_{\ell_{m-1},\bm A}\\
		=&\frac{N}{m\left(m+1\right)}+\hat{r}_{\ell_{m},\bm A}-\hat{r}_{\ell_{m-1},\bm A}.
		\end{split}
		\end{align}
		From the Equations \ref{eq5} and \ref{eq4} $\hat{d}_{m}$ is given by
		\begin{align}
		\label{eq6}
		\begin{split}
		\hat{d_{m}}=&\alert{\frac{1}{m+1}N-\hat{r}_{\ell_{m-1},\bm A}}\\\stackrel{(c)}{=}&\frac{1}{m+1}N-\frac{(m-2)}{(m-1)m}N\\
		=&\frac{2N}{(m-1)m(m+1)}.\end{split}
		\end{align} Note that in $(c)$, we used  $\hat{r}_{\ell_{m},\bm A}=\frac{(m-1)}{m(m+1)}N$.
		Considering Eq. \eqref{eq6} and the approximation $\ell_{\max}\approx\frac{N}{2}$, $\hat{\ell}_{m}$ corresponds to (cf. Fig. \ref{fig:non_integer})
		\begin{align}
		\notag \hat{\ell}_{m}=&\frac{N}{2}-\sum_{i=2}^{m}\hat{d}_{i}\\
		\notag =&\frac{N}{2}-N\sum_{i=2}^{m}\frac{2}{(i-1)i(i+1)}\\
		\notag =&\frac{N}{2}-N\left(\sum_{i=2}^{m}\frac{1}{i-1}\alert{-\frac{2}{i}}+\frac{1}{i+1}\right)\\
		=&\alert{\frac{N}{2}-N\left[\left(\sum_{i=2}^{m}\frac{1}{i-1}-\frac{1}{i}\right)+\left(\sum_{i=2}^{m}-\frac{1}{i}+\frac{1}{i+1}\right)\right]}\nonumber \\
		=&\alert{\frac{N}{2}-N\left[1-\frac{1}{m}-\frac{1}{2}+\frac{1}{m+1}\right]}\nonumber \\
		\label{eq7} =&\frac{N}{m(m+1)}.
		\end{align}
		Eq. \eqref{eq7} represents an estimate on the number of colluding servers $\hat{\ell}_{m}$ at which $\hat{r}_{\ell,\bm B}$ increases to $m$. We use Eq. \eqref{eq7} to determine $m$ as a function of $N$ and $\ell$. Specifically, recall that for any \textcolor{black}{$\ell\in(\hat{\ell}_{m+1},\hat{\ell}_{m}]$, we know that $\hat{r}_{\ell, \bm B}$ is constant and we} need to ensure that
		\begin{align*}
		\ell\geq \hat{\ell}_{m+1},
		\end{align*} or according to \eqref{eq7} equivalently
		\begin{align}
		\notag
		&\ell\geq\frac{N}{(m+2)(m+1)}\\
		\label{eq10}
		\Leftrightarrow\; & m^2 +3m-\frac{N}{\ell}+\alert{2}\geq 0.
		\end{align} Due to Lemma \ref{claim1}, we choose the smallest $m$ that satisfies \eqref{eq10}. This gives us 
		\begin{align}
		\label{r_B}
		 \textcolor{black}{\hat{r}_{\ell,\bm B}=\biggl\lceil-\frac{3}{2}+\sqrt{\frac{1}{4}+\frac{N}{\ell}}\:\biggr\rceil}
		\end{align} for \textcolor{black}{$\ell\in(\hat{\ell}_{m+1},\hat{\ell}_{m}]$ with $\ell\in\mathbb{Z}^{+}$}. \textcolor{black}{By using the ceiling function, the values of $\hat{r}_{\ell,\bm B}$ which do not correspond to integer numbers are removed for integers $\ell$. This concludes the proof of Theorem \ref{th:res}.}

\ifCLASSOPTIONcaptionsoff
  \newpage
\fi

\bibliography{Citations}

\end{document}